\documentclass{LMCS}

\def\dOi{10(4:9)2014}
\lmcsheading%
{\dOi}
{1--29}
{}
{}
{Feb.~18, 2014}
{Dec.~10, 2014}
{}

\ACMCCS{[{\bf Software and its engineering}]: Software notations and
  tools---General programming languages---Language features; [{\bf
      Theory of computation}]: Semantics and reasoning}

\usepackage{mathtools, amssymb}
\usepackage{xspace}
\usepackage{mathpartir}
\usepackage[llbracket,rrbracket]{stmaryrd}
\usepackage{xypic}
\usepackage{booktabs}
\usepackage{listings}

\newcommand{\NN}{\mathbb{N}} % natural numbers

\newcommand{\Eff}{\textit{Eff}\xspace}
\newcommand{\set}[1]{\{#1\}}
\newcommand{\such}{\mid}

\newcommand{\lift}[1]{#1^\dagger}

\newcommand{\stricto}{\multimap}

\newcommand{\effcase}{\mathop{\text{\texttt{|}}}}

\newcommand{\bnfis}{\mathrel{\;{:}{:}\!=}\;}
\newcommand{\bnfor}{\mathrel{\;\big|\;}}
\newcommand{\defeq}{\mathrel{\;\stackrel{\text{def}}{=}\;}}
\newcommand{\rulename}[1]{{\mdseries \small \textsc{#1}}}
\newcommand{\ctx}{\Gamma}
\newcommand{\ent}{\vdash}
\newcommand{\T}{\mathrel{:}}
\newcommand{\E}{\,{!}\,}
\newcommand{\sig}{\Sigma}

\renewcommand{\le}{\leqslant}
\renewcommand{\leq}{\leqslant}

\newcommand{\type}[1]{\mathtt{#1}}

\newcommand{\emptyty}{\type{empty}}
\newcommand{\unitty}{\type{unit}}
\newcommand{\natty}{\type{nat}}
\newcommand{\boolty}{\type{bool}}
\renewcommand{\to}[1][]{
  \def\arg{#1}%
  \ifx\arg\emptyarg\rightarrow\else\xrightarrow{#1}\fi%
}
\newcommand{\hto}[1][]{
  \def\arg{#1}%
  \ifx\arg\emptyarg\Rightarrow\else\xRightarrow{#1}\fi%
}
\newcommand{\C}{\underline{C}}
\newcommand{\D}{\underline{D}}

\newcommand{\purely}{\varnothing}

\renewcommand{\succ}{\kpre{succ}}

\newcommand{\op}{\mathtt{op}}
\newcommand{\hash}[2]{#1\text{\texttt{\char35}}#2}
\newcommand{\call}[4]{{\hash{#1}{#2}\,{#3}\, #4}}
\newcommand{\cont}[3][]{#1(#2.\,#3#1)}  
\newcommand{\anon}{\text{\texttt{\char95}}\,}

\newcommand{\kword}[1]{\ensuremath{\mathtt{#1}}}
\newcommand{\kord}[1]{\kword{#1}}
\newcommand{\kop}[1]{\;\kword{#1}\;}
\newcommand{\kpre}[1]{\kword{#1}\;}

\newcommand{\unt}{\kord{\text{\textnormal{\texttt{()}}}}}

\newcommand{\absurd}{\kpre{absurd}}

\newcommand{\matchnat}[4]{\kpre{match} #1 \kop{with} 0 \mapsto #2 \effcase \succ #3 \mapsto #4}

\newcommand{\tru}{\kord{true}}
\newcommand{\fls}{\kord{false}}

\newcommand{\semtru}{\mathrm{t\!t}}
\newcommand{\semfls}{\mathrm{ff}}

\newcommand{\fun}[1]{\kpre{fun} #1 \mapsto}
\newcommand{\handler}{\kpre{handler}}
\newcommand{\ocs}{\textit{ocs}}
\newcommand{\ocsnil}{\textit{nil}}
\newcommand{\val}{\kpre{val}}

\newcommand{\conditional}[3]{\kpre{if} #1 \kop{then} #2 \kop{else} #3}
\newcommand{\withhandle}[2]{\kpre{with} #1 \kop{handle} #2}

\newcommand{\letin}[1]{\kpre{let} #1 \kop{in}}
\newcommand{\letrecin}[1]{\kpre{let} \kpre{rec} #1 \kop{in}}

\newcommand{\sem}[2]{\llbracket #1 \rrbracket #2}
\newcommand{\xsem}[1]{\llbracket #1 \rrbracket}
\newcommand{\tsem}[1]{\llbracket\![ #1 \rrbracket\!]}
\newcommand{\skel}[1]{\llbracket #1 \rrbracket}
\newcommand{\skele}[1]{\llbracket #1 \rrbracket_e}
\newcommand{\skelc}[1]{\llbracket #1 \rrbracket_c}

\newcommand{\retract}[4]{\xymatrix{**[l]{#1} \ar@<0.25em>[r]^{#3} & **[r]{#2} \ar@<0.25em>[l]^{#4}}}
\newcommand{\retractx}[4]{\xymatrix@-3em{**[l]{#1} \ar@<0.25em>[r]^{#3} & **[r]{#2} \ar@<0.25em>[l]^{#4}}}
\newcommand{\retractxx}[4]{\xymatrix@-6em{**[l]{#1} \ar@<0.25em>[r]^{#3} & **[r]{#2} \ar@<0.25em>[l]^{#4}}}

\newcommand{\compdom}{T}

\newcommand{\inval}{\mathsf{in}_{\mathsf{val}}}
\newcommand{\inop}[1]{\mathsf{in}_{#1}}

\newcommand{\aprx}[1]{\triangleleft_{#1}}

\newcommand{\step}{\leadsto}
\newcommand{\eval}{\Downarrow}

\makeatletter
\providecommand*{\cupdot}{%
  \mathbin{%
    \mathpalette\@cupdot{}%
  }%
}
\newcommand*{\@cupdot}[2]{%
  \ooalign{%
    $\m@th#1\cup$\cr
    \sbox0{$#1\cup$}%
    \dimen@=\ht0 %
    \sbox0{$\m@th#1\cdot$}%
    \advance\dimen@ by -\ht0 %
    \dimen@=.5\dimen@
    \hidewidth\raise\dimen@\box0\hidewidth
  }%
}

\providecommand*{\bigcupdot}{%
  \mathop{%
    \vphantom{\bigcup}%
    \mathpalette\@bigcupdot{}%
  }%
}
\newcommand*{\@bigcupdot}[2]{%
  \ooalign{%
    $\m@th#1\bigcup$\cr
    \sbox0{$#1\bigcup$}%
    \dimen@=\ht0 %
    \advance\dimen@ by -\dp0 %
    \sbox0{\scalebox{2}{$\m@th#1\cdot$}}%
    \advance\dimen@ by -\ht0 %
    \dimen@=.5\dimen@
    \hidewidth\raise\dimen@\box0\hidewidth
  }%
}
\makeatother

\newcommand{\Drt}{\Delta}
\newcommand{\Rgn}{R}
\newcommand{\ops}{\mathcal{O}}
\newcommand{\iotas}{\mathcal{I}}
\newcommand{\exs}[1]{\exists #1 .\ }
\newcommand{\fra}[1]{\forall #1 .\ }
\newcommand{\evctx}{\mathcal{E}}
\newcommand{\Equiv}{\;\;\equiv\;\;}
\newcommand{\hndl}{\mathcal{H}}

\lstnewenvironment{source}{\lstset{
basicstyle=\ttfamily\small,
}}{}

\usepackage{hyperref}
\begin{document}

\title[An Effect System for Algebraic Effects and Handlers]{An Effect
  System for Algebraic Effects and Handlers\rsuper*}

\author[A.~Bauer]{Andrej Bauer}
\address{Faculty of Mathematics and Physics, University of Ljubljana, Slovenia}
\email{Andrej.Bauer@andrej.com, matija.pretnar@fmf.uni-lj.si}

\author[M.~Pretnar]{Matija Pretnar}
%\address{Faculty of Mathematics and Physics, University of Ljubljana, Slovenia}
%\email{matija.pretnar@fmf.uni-lj.si}


\titlecomment{{\lsuper*}A preliminary version of this work was presented at CALCO 2013, see~\cite{bauer2013effect}.}

\begin{abstract}
  We present an effect system for \emph{core \Eff}, a simplified variant of \Eff,
  which is an ML-style programming language with first-class algebraic effects and handlers. We
  define an expressive effect system and prove safety of operational semantics with
  respect to it. Then we give a domain-theoretic denotational semantics of core \Eff,
  using Pitts's theory of minimal invariant relations, and prove
  it adequate. We use this fact to develop tools for finding useful contextual equivalences,
  including an induction principle.
  To demonstrate their usefulness, we use these tools to derive the usual equations
  for mutable state, including a general commutativity law for computations using non-interfering references.
  We have formalized the effect system, the operational
  semantics, and the safety theorem in Twelf.
\end{abstract}

\maketitle

\section{Introduction}
\label{sec:introduction}

An \emph{effect system} supplements a traditional type system for a programming language
with information about which computational effects may, will, or will not happen when a
piece of code is executed. A well designed and solidly implemented effect system helps
programmers understand source code, find mistakes, as well as safely
rearrange, optimize, and parallelize code~\cite{lucassen88polymorphic,kammar12algebraic}.
As many before
us~\cite{lucassen88polymorphic,talpin1992polymorphic,wadler1999marriage,kammar13handlers}
we take on the task of striking just the right balance between simplicity and expressiveness
by devising an effect system for \Eff~\cite{bauer12programming},
an ML-style programming language with first-class algebraic effects~\cite{plotkin03algebraic,plotkin2001adequacy}
and handlers~\cite{plotkin13handling}.

Our effect system is \emph{descriptive} in the sense that it provides information about possible computational effects but it does not prescribe them. In contrast, Haskell's monads \emph{prescribe} the possible effects by wrapping types into computational monads. In the implementation we envision effect inference which never fails, although in some cases it may be uninformative. Of course, typing errors are still errors.

An important feature of our effect system is \emph{non-monotonicity}: it detects the fact that a handler removes some effects. For instance, a piece of code which uses mutable state is determined to actually be pure when wrapped by a handler that handles away lookups and updates.

\goodbreak
Our contributions are as follows:
\begin{enumerate}
\item 
  We define \emph{core Eff}, a fragment of the language which retains the essential
  features of \Eff, including first-class handlers and instances (Section~\ref{sec:core-eff}),
  although we leave out dynamic creation of new instances.
\item
  We give small-step and big-step operational semantics for core \Eff and
  show them to be equivalent
  (Section~\ref{sec:operational-semantics}).
\item
  We devise an expressive effect system for core \Eff and prove
  safety of the operational semantics with respect to it (Section~\ref{sec:effect-system}).
\item
  Using the standard domain-theoretic apparatus
  and Pitts's theory of minimal invariant relations~\cite{Pitts96},
  we provide denotational semantics for core \Eff and
  prove an adequacy theorem (Section~\ref{sec:denotational-semantics}).
\item
  We identify a set of observational equivalences and an induction principle
  that allow us to reason about effectful computations (Section~\ref{sec:reasoning}).
\item 
  We demonstrate how the equivalences are used by deriving the standard equations for
  state from general principles. The induction principle is used in a proof of a general
  commutativity law which allows us to interchange two computations that use
  non-interfering references (Section~\ref{sec:examples}).
\item
  We formalized
  core \Eff, the operational semantics, the effect system, and the safety theorem
  in Twelf~\cite{twelf} (Section~\ref{sec:formalization}).
\end{enumerate}

%%% Local Variables: 
%%% mode: latex
%%% TeX-master: "effect-system"
%%% End: 

\section{Core Eff}
\label{sec:core-eff}

The current implementation of \Eff includes a number of features, such as syntactic sugar, products, records, inductive types, type definitions, effect definitions, etc., which are inessential for a conceptual analysis. We therefore restrict attention to \emph{core \Eff}, a fragment of the language described here.
We refer the readers to~\cite{bauer12programming} for a more thorough introduction of how
one actually programs in \Eff.

In \Eff all computational effects are accessed uniformly and exclusively through \emph{operations}. These are a primitive concept, of which typical examples are reading and writing on a communication channel, updating and looking up the contents of a reference, and raising an exception. Thus, in \Eff each terminating computation results either in an effect-free value, or it calls an operation. Each operation has an associated delimited \emph{continuation}, which is a suspended computation awaiting the result of the operation.

Operations do not actually perform effects, but are just suspended computations 
whose behavior is controlled by a second primitive notion, the \emph{effect handlers}. These are like exception handlers, except that an effect handler has access to the continuation of the handled operation, and so may restart the computation after the operation is handled. With handlers we may implement all the usual computational effects, as well as great variety of others, such as transactional memory, non-deterministic execution strategies, stream redirection, cooperative multi-threading, and delimited continuations. At the top level there may be built-in handlers that provide interaction with the external environment, although we do not consider these in core \Eff.

Since \Eff is geared towards practical programming, it and core \Eff depart in several respects from previous work on handlers and algebraic effects~\cite{plotkin13handling, kammar13handlers}. First, rather than imposing equations on handlers by a typing discipline, the programmer may write arbitrary handlers, and then prove that a particular handler satisfies the desired equations. We demonstrate the technique in Section~\ref{sub:reasoning-references}, where we implement a state handler and show that it satisfies the standard equations.
Second, \Eff uses fine-grained call-by-value evaluation strategy~\cite{levy03modelling} rather than the theoretically more desirable call-by-push-value~\cite{levy06call-by-push-value} because we found the former to be closer to programming practice as well as easier to implement.
Third, every effect has multiple \emph{instances}. For example, a program may write and read from multiple communication channels, raise different kinds of exceptions, and manipulate multipartite state. Thus in core \Eff an \emph{operation symbol $\op$} is always paired with an \emph{instance $\iota$} to give an \emph{operation $\hash{\iota}{\op}$}. From a theoretical point of view instances are straightforward (as long as we do not generate them dynamically) and inessential, but are absolutely necessary for practical programming.

\subsection{Effects and types}
\label{sec:effects-and-types}

To get things going we presume given a collection of \emph{effects}
\begin{equation*}
  \text{Effect}\ E \bnfis \kord{exception} \bnfor \kord{ref} \bnfor \cdots
\end{equation*}
which in full \Eff are declared by the programmer. For each $E$ there is a given set $\iotas_E$ of \emph{instances} $\iota_1, \iota_2, \iota_3, \ldots$ of $E$. The instances may be thought of as atomic names. In full \Eff they may be dynamically created, but to keep the semantics reasonably simple we assume a fixed set. Additionally, with each $E$ we associate a set of $\ops_E$ of \emph{operation symbols} $\op_1, \op_2, \ldots$ An operation symbol is associated with at most one effect.

The terms of core \Eff are split into effect-free \emph{expressions} and possibly effectful \emph{computations}, as described in the next subsection.
Consequently, the type system of core \Eff consists of \emph{pure types} for expressions and \emph{dirty types} for computations:
\begin{align*}
  \text{Pure type}\ A, B &\bnfis
    \boolty \bnfor
    \natty \bnfor
    \unitty \bnfor
    \emptyty \bnfor
    A \to \C \bnfor
    E^\Rgn \bnfor
    \C \hto \D
    \\
  \text{Dirty type}~\C, \D &\bnfis
    A \E \Drt
  \\
  \text{Region}\ \Rgn &\bnfis \set{\iota_1, \dots, \iota_n}
  \\
  \text{Dirt}\ \Drt &\bnfis \set{\hash{\iota_1}{\op_1}, \dots, \hash{\iota_n}{\op_n}}
\end{align*}
A dirty type $A \E \Drt$ is just a pure type $A$ tagged with a finite set $\Drt$ of operations that might be called during evaluation. We require that any operation $\hash{\iota}{\op}$ appearing in $\Drt$ is well-formed in the sense that $\iota \in \iotas_E$ and $\op \in \ops_E$ for some effect~$E$.

The pure types comprise the usual ground types, the function types $A \to \C$, the effect types $E^\Rgn$, and the \emph{handler types $\C \hto \D$}. Note that the function type takes pure types to dirty ones because a function accepts a pure expression as an argument and may call operations when evaluated.  We let $\E$ bind more strongly than $\to$, so that $A \to B \E \Drt$ means $A \to (B \E \Drt)$. Each effect type $E^\Rgn$ is tagged with a finite set of instances $\Rgn = \set{\iota_1, \dots, \iota_n} \subseteq \iotas_E$ which tells us that the expression equals one of the instances in~$\Rgn$. Finally, $\C \hto \D$ is the type of handlers which take computations of \emph{ingoing} type $\C$ to computations of \emph{outgoing} type $\D$.

We assume that each effect $E$ has an associated \emph{effect signature}
\begin{equation*}
  \sig_E = 
  \set{\op_1 \T A^{\op_1} \to B^{\op_1}, \dots, \op_n \T A^{\op_n} \to B^{\op_n}}
\end{equation*}
which assigns to each operation $\op_i \in \ops_E$ its \emph{parameter type $A^{\op_i}$} and \emph{result type
  $B^{\op_i}$}. In full \Eff the signature is part of the definition of an effect, so
for instance we might have
\begin{align*}
  \sig_{\kord{ref}} &= \set{\kord{lookup} \T \unitty \to \natty, \kord{update} \T \natty \to \unitty}, \\
  \sig_{\kord{exception}} &= \set{\kord{raise} \T \unitty \to \emptyty}.
\end{align*}
Note that the signature may create circularities such as
\begin{equation*}
  \sig_E = \set{ \op \T \unitty \to (\unitty \to \unitty \E \set{\hash{\iota}{\op}}) }.
\end{equation*}
Consequently, the denotational semantics of types in Section~\ref{sec:denotational-semantics}
will involve recursive domain equations.

\subsection{Terms}
\label{sec:terms}

The abstract syntax of terms of core \Eff is as follows:
\begin{align*}
  \text{Expression}~e \bnfis {}
    &x \bnfor
    \tru \bnfor
    \fls \bnfor
    0 \bnfor
    \succ e \bnfor
    \unt \bnfor
    \fun{x \T A} c \bnfor
    \iota \bnfor
    h
    \\  
  \text{Handler}~h \bnfis {}
    &(
      \handler
      \val x \T A \mapsto c_v \effcase
      \ocs
    )
    \\
  \text{Operation cases}~\ocs \bnfis {}
    &\ocsnil_{\C} \bnfor
    (\call{e}{\op}{x}{k} \mapsto c \effcase \ocs)
    \\
  \text{Computation}~c \bnfis
    & \val e \bnfor
    \call{e_1}{\op}{e_2}{\cont{y}{c}} \bnfor
    \withhandle{e}{c} \bnfor \\
    & \conditional{e}{c_1}{c_2} \bnfor
    \kord{absurd}_{\C} \; e \bnfor
    e_1 \, e_2 \bnfor \\
    & (\matchnat{e}{c_1}{x}{c_2}) \bnfor \\
    &\letin{x = c_1} c_2 \bnfor
    \letrecin{f \, x \T A \to \C = c_1} c_2
\end{align*}
In order to ensure that each term has at most one skeletal typing derivation, cf.\ Subsection~\ref{sub:skeletal-types}, certain terms include typing annotations.
We shall omit these when they do not play a role.
The \Eff implementation does not have typing annotations because its effect system
automatically infers types and effects~\cite{Pretnar13}.

An expression is either a variable, a constant of ground type, a function abstraction (note that we abstract over computations), an \emph{effect instance}, or a \emph{handler}. It is worth noting that both instances and handlers are first-class values. We sometimes abbreviate $\kord{succ}^n\, 0$ as $n$.
A handler consists of a single \emph{value case} and multiple \emph{operation cases}, which describe
how values and operations are handled, respectively.
We defined operation cases inductively as lists, which is how they are formalized in Twelf,
but we also write them as $(\hash{e_i}{\op_i} \, x_i \, k_i \mapsto c_i)_i$.

A computation is either a pure expression, an \emph{operation call}, a $\kord{handle}$ construct, an eliminator for a ground type, an application, a $\kord{let}$ binding, or a recursive function definition.

%%% Local Variables: 
%%% mode: latex
%%% TeX-master: "effect-system"
%%% End: 

\section{Operational semantics}
\label{sec:operational-semantics}

We first describe the operational semantics informally.
A computation $\val e$ is pure and indicates a ``final'' result~$e$, while an operation call $\call{e_1}{\op}{e_2}{\cont{y}{c}}$ is the principal way of triggering an effect. The instance $e_1$ and the operation symbol $\op$ together form an operation $\hash{e_1}{\op}$, its parameter is $e_2$, and $\cont{y}{c}$ the delimited continuation. We do not expect programmers to write explicit continuations, so the concrete syntax of \Eff only gives access to calls through functions of the form $\fun{x} \call{e}{\op}{x}{\cont{y}{\val y}}$, known also as \emph{generic effects}~\cite{plotkin03algebraic}.
In examples we shall use generic effects rather than explicit continuations, and there we write them as $\hash{e}{\op}$.
A general operation call $\call{e_1}{\op}{e_2}{\cont{y}{c}}$ may then be expressed in terms of a generic effect and a $\kord{let}$ binding as $\letin{y = \hash{e_1}{\op} \, e_2} c$.

A binding $\letin{x = c_1} c_2$ is evaluated as follows:
\begin{enumerate}
\item If $c_1$ evaluates to $\val e$ then the binding evaluates as $c_2$ with $x$ bound to $e$.
\item If $c_1$ evaluates to an operation call $\call{\iota}{\op}{e}{\cont{y}{c_1'}}$, then the binding evaluates to \[\call{\iota}{\op}{e}{\cont{y}{\letin{x = c_1'} c_2}},\] where we assume that $y$ does not occur free in $c_2$.
\end{enumerate}
It may be helpful to think of $\kord{val}$ and $\kord{let}$ as being similar to Haskell $\kord{return}$ and $\kord{do}$, respectively. In ML $\kord{val}$ is invisible, while $\kord{let}$ is essentially the same as ours.

The $\kord{handle}$ construct applies a handler to a computation. If $h$ is the \emph{handler}
\begin{equation*}
  \handler \val x \mapsto c_v \effcase (\hash{\iota_i}{\op_i} \, x_i \, k_i \mapsto c_i)_i
\end{equation*}
and $c$ is a computation, then $\withhandle{h}{c}$ first evaluates $c$ which is then handled according to $h$:
\begin{enumerate}
\item If $c$ evaluates to $\val e$, then the $\kord{handle}$ construct evaluates as $c_v$ with $x$ bound to~$e$.
\item If $c$ evaluates to $\call{\iota}{\op}{e'}{\cont{y}{c'}}$, and $\hash{\iota_i}{\op_i} \, x_i \, k_i \mapsto c_i$ is the first operation case in $h$ for which $\hash{\iota}{\op} = \hash{\iota_i}{\op_i}$ then the $\kord{handle}$ construct evaluates to $c_i$ with $x_i$ and $k_i$ bound to $e'$ and $\fun{y} \withhandle{h}{c'}$, respectively. We assume that $y$ does not occur free in $h$.
\item If $c$ evaluates to an operation call $\call{\iota}{\op}{e'}{\cont{y}{c'}}$ which is not listed by~$h$, then the $\kord{handle}$ construct propagates the call and acts as if $h$ contained the clause
  \begin{equation*}
    \hash{\iota}{\op} \, x \, k \mapsto \call{\iota}{\op}{x}{\cont{y}{k \, y}}.
  \end{equation*}
  Thus it evaluates to $\call{\iota}{\op}{e'}{\cont{y}{\withhandle{h}{c'}}}$, where again we assume that $y$ does not occur free in~$h$.
\end{enumerate}
Note that the handler always wraps itself around the continuation so that subsequent
operations are handled as well. A binding
$
  \letin{x = c_1}{c_2}
$
is equivalent to
\begin{equation*}
  \withhandle{(\handler \val x \mapsto c_2)}{c_1}
\end{equation*}
so we could theoretically omit $\kord{let}$.

\subsection{Small-step semantics}
\label{sec:small-step-semantics}

The small-step operational semantics of core \Eff is defined in terms of a relation $c \step c'$,
which intuitively means that the computation $c$ takes a single step to $c'$.
There is no operational semantics for expressions, which are just inert pieces of data.
The relation $\step$ is defined inductively by the following rules:
{\small
  \begin{mathpar}
  \inferrule{
  }{
    \conditional{\tru}{c_1}{c_2} \step c_1
  }

  \inferrule{
  }{
    \conditional{\fls}{c_1}{c_2} \step c_2
  }

  \inferrule{
  }{
    (\matchnat{0}{c_1}{x}{c_2}) \step c_1
  }
  \hfill
  \inferrule{
  }{
    (\matchnat{(\succ e)}{c_1}{x}{c_2}) \step c_2[e / x]
  }

  \inferrule{
  }{
    (\fun{x} c) \, e \step c[e / x]
  }

  \inferrule{
    c_1 \step c_1'
  }{
    \letin{x = c_1} c_2 \step \letin{x = c_1'} c_2
  }

  \inferrule{
  }{
    \letin{x = (\val e)} c \step c[e / x]
  }

  \inferrule{
  }{
    \letin{x = (\call{\iota}{\op}{e}{\cont{y}{c_1}})} c_2
      \step \call{\iota}{\op}{e}{\cont{y}{\letin{x = c_1} c_2}}      
  }

  \inferrule{
  }{
    \letrecin{f \, x = c_1} c_2
      \step c_2[(\fun{x} \letrecin{f \, x = c_1} c_1) / f]
  }

  \inferrule{
    c \step c'
  }{
    \withhandle{e}{c} \step \withhandle{e}{c'}
  }
  \hfill
  \inferrule{
  }{
    \withhandle{(\handler \val x \mapsto c_v \effcase \ocs)}{(\val e)} \step c_v[e / x]
  }

  \inferrule{
    h \defeq (\handler \val x \mapsto c_v \effcase \ocs) \\
    (\op \T A^\op \to B^\op) \in \sig_E
  }{
    \withhandle{h}{(\call{\iota}{\op}{e}{\cont{y}{c}})}
      \step \ocs_{\hash{\iota}{\op}}(e, (\fun{y \T B^\op} \withhandle{h}{c}))
  }
\end{mathpar}}%
In the last rule for $\kord{let}$ binding and the last rule for the $\kord{handle}$ construct variable $y$ must not occur free in $c_2$ and $h$ respectively, and it goes without saying that the substitutions are capture avoiding. In the last rule we have an auxiliary definition of
$\ocs_{\hash{\iota}{\op}}$:
\begin{align*}
  (\ocsnil)_{\hash{\iota}{\op}}(e, \kappa) &= \call{\iota}{\op}{e}{\cont{y}{\kappa \, y}} \\
  (\call{\iota'}{\op'}{x}{k} \mapsto c \effcase \ocs)_{\hash{\iota}{\op}}(e, \kappa) &=
  \begin{cases}
    c[e / x, \kappa / k] & \text{if $\hash{\iota}{\op} = \hash{\iota'}{\op'}$} \\
    \ocs_{\hash{\iota}{op}}(e, \kappa) & \text{otherwise}
  \end{cases}
\end{align*}
In words, $\ocs_{\hash{\iota}{\op}}(e, \kappa)$ finds the first handler case in $\ocs$ that matches the operation $\hash{\iota}{\op}$ and executes it, or calls the operation again if no match is found.

\begin{exa}
\label{exa:reduction}
The non-standard state handler (recall that in examples we use generic effects)
\begin{align*}
  h \defeq {} &\handler \\
  &\effcase \val x \mapsto \hash{\iota}{\kord{update}} \, x \\
  &\effcase \call{\iota}{\kord{lookup}}{x}{k} \mapsto k \, 1 \\
  &\effcase \call{\iota}{\kord{update}}{x}{k} \mapsto k \, \unt
\end{align*}
treats the reference $\iota$ as if its content were always $1$, and
updates $\iota$ with the final result of the handled computation.
This update is not handled by~$h$ because it escapes its scope.
If we use $h$ to handle the computation
\begin{align*}
  c \defeq {}
  &\letin{x_1 = \hash{\iota}{\kord{lookup}} \, \unt} \\
  &\letin{x_2 = \hash{\iota}{\kord{update}} \, x_1} \\
  &\val (\succ x_1)
\end{align*}
the outcome of the first lookup is $1$, which is bound to $x_1$, the update is ignored and
finally $\hash{\iota}{\kord{update}} \, 2$ is called.
The exact reduction sequence is as follows, where we underline the active parts at each step
and indicate desugaring of generic effects with $\equiv$:
{\small
\begin{align*}
  &\withhandle{h}{\\
  &\quad \letin{x_1 = \underline{\hash{\iota}{\kord{lookup}} \, \unt}} \letin{x_2 = \hash{\iota}{\kord{update}} \, x_1} \val (\succ x_1)} \equiv {} \\
  &\withhandle{h}{\\
  &\quad \letin{\underline{x_1} = \underline{\hash{\iota}{\kord{lookup}}{\unt}{\cont{y_1}{\val y_1}}}} \letin{x_2 = \hash{\iota}{\kord{update}} \, x_1} \val (\succ x_1)} \step \\
  &\withhandle{\underline{h}}{\\
  &\quad \underline{\hash{\iota}{\kord{lookup}}{\unt}{\cont{y_1}{\letin{x_1 = \val y_1} \letin{x_2 = \hash{\iota}{\kord{update}} \, x_1} \val (\succ x_1)}}}} \step \\
  &\underline{(\fun{y_1} \withhandle{h}{(\letin{x_1 = \val y_1} \letin{x_2 = \hash{\iota}{\kord{update}} \, x_1} \val (\succ x_1)))}} \,\, \underline{1\vphantom{()}} \step \\
  &\withhandle{h}{(\letin{\underline{x_1} = \underline{\val 1}} \letin{x_2 = \hash{\iota}{\kord{update}} \, x_1} \val (\succ x_1))} \step \\
  &\withhandle{h}{(\letin{x_2 = \underline{\hash{\iota}{\kord{update}} \, 1}} \val 2)} \equiv \\
  &\withhandle{h}{(\letin{\underline{x_2} = \underline{\call{\iota}{\kord{update}}{1}{\cont{y_2}{\val y_2}}}} \val 2)} \step \\
  &\withhandle{\underline{h}}{(\underline{\call{\iota}{\kord{update}}{1}{\cont{y_2}{\letin{x_2 = \val y_2} \val 2}}})} \step \\
  &\underline{(\fun{y_2} \withhandle{h}{(\letin{x_2 = \val y_2} \val 2)})} \,\, \underline{\unt\vphantom{()}} \step \\
  &\withhandle{h}{(\letin{\underline{x_2} = \underline{\val \unt}} \val 2)} \step \\
  &\withhandle{\underline{h}}{\underline{(\val 2)}} \step
  \underline{\hash{\iota}{\kord{update}} \, 2} \equiv \call{\iota}{\kord{update}}{2}{\cont{y_3}{\val y_3}}
\end{align*}}%
\end{exa}

\subsection{Big-step semantics}
\label{sec:big-step-semantics}

In addition to small-step operational semantics, we also provide a big-step variant, which is closer to the actual implementation of \Eff. Define a \emph{result} to be a pure expression or an operation call:
\[
  \text{Result}\ r \bnfis
  \val e \bnfor
  \call{\iota}{\op}{e}{\cont{x}{c}}
\]
Big-step semantics $c \eval r$ evaluates a computation $c$ to a result $r$, according
to the following inductive rules:
{\small\begin{mathpar}
  \inferrule{
    c_1 \eval r
  }{
    \conditional{\tru}{c_1}{c_2} \eval r
  }

  \inferrule{
    c_2 \eval r
  }{
    \conditional{\fls}{c_1}{c_2} \eval r
  }

  \inferrule{
    c_1 \eval r
  }{
    (\matchnat{0}{c_1}{x}{c_2}) \eval r
  }

  \inferrule{
    c_2[e / x] \eval r
  }{
    (\matchnat{\succ e}{c_1}{x}{c_2}) \eval r
  }

  \inferrule{
    c[e / x] \eval r
  }{
    (\fun{x} c) \, e \eval r
  }

  \inferrule{
  }{
    \val e \eval \val e
  }

  \inferrule{
  }{
    \call{\iota}{\op}{e}{\cont{x}{c}} \eval \call{\iota}{\op}{e}{\cont{x}{c}}
  }

  \inferrule{
    c_1 \eval \val e \\
    c_2[e / x] \eval r
  }{
    \letin{x = c_1} c_2 \eval r
  }

  \inferrule{
    c_1 \eval \call{\iota}{\op}{e}{\cont{y}{c}}
  }{
    \letin{x = c_1} c_2 \eval \call{\iota}{\op}{e}{\cont{y}{\letin{x = c} c_2}}
  }

  \inferrule{
    c_2[(\fun{x} \letrecin{f \, x = c_1} c_1) / f] \eval r
  }{
    \letrecin{f \, x = c_1} c_2 \eval r
  }

  \inferrule{
    c \eval \val e \\
    c_v[e / x] \eval r
  }{
    \withhandle{(\handler \val x \mapsto c_v \effcase \ocs)}{c} \eval r
  }

  \inferrule{
    h \defeq (\handler \val x \mapsto c_v \effcase \ocs) \\
    c \eval \call{\iota}{\op}{e}{\cont{y}{c}} \\
    (\op \T A^\op \to B^\op) \in \sig_E \\\\
    \ocs_{\hash{\iota}{\op}}(e, (\fun{y \T B^\op} \withhandle{h}{c})) \eval r
  }{
    \withhandle{h}{c} \eval r
  }
\end{mathpar}}%
To relate the two semantics we define an auxiliary relation $\step^*$ by the rules
\begin{mathpar}
  \inferrule{
  }{
    \val e \step^* \val e
  }

  \inferrule{
  }{
    \call{\iota}{\op}{e}{\cont{x}{c}} \step^* \call{\iota}{\op}{e}{\cont{x}{c}}
  }

  \inferrule{
    c \step c' \\
    c' \step^* r
  }{
    c \step^* r
  }
\end{mathpar}
This is roughly the reflexive transitive closure of $\step$, except that it relates
computations to results rather than to computations. The small-step and big-step semantics
agree in the following sense.

\begin{prop}
  For all computations~$c$ and results~$r$, $c \eval r$ if and only if $c \step^* r$.
\end{prop}

\begin{proof}
  Both directions of the equivalence proceed by a routine induction. The formalized proofs
  of the two implications can be found in the file \verb|small-big.elf|.
\end{proof}

%%% Local Variables: 
%%% mode: latex
%%% TeX-master: "effect-system"
%%% End: 

\section{An effect system}
\label{sec:effect-system}

\subsection{Subtyping}
\label{sub:subtyping}

As in most effect systems,
we need to take care of the \emph{poisoning problem}~\cite{wansbrough1999once}.
For example, what should be the type of~$\mathit{ignore}$ in
\begin{align*}
  &\letin{\mathit{ignore} = \val (\fun{msg} \val \unt)} \\
  &\letin{f = \conditional{b}{(\val ignore)}{(\val \hash{\kord{std}}{\kord{write}}})} \\
  &\val \mathit{ignore}
\end{align*}
assuming we have the ground type $\type{string}$ and a boolean expression $b$?
If we give it the desired type~$\type{string} \to \unitty \E \purely$ then there is a type mismatch between the branches in the conditional statement, whereas the dirty type $\type{string} \to \unitty \E \set{\hash{\kord{std}}{\kord{write}}}$ loses the valuable knowledge that $\mathit{ignore}$ is a pure function. The simplest antidote to the poisoning problem is subtyping so that $\mathit{ignore}$ may be given the function type with empty dirt which is coerced in the conditional statement to a supertype that matches the other branch.

For our purposes, a straightforward variant of \emph{structural} subtyping~\cite{fuh1990type} suffices. We have subtyping of pure types $A \leq A'$ and of dirty types $\C \leq \C$, given by the rules
{\small\begin{mathpar}
  \inferrule{
  }{
    \boolty \le \boolty
  }

  \inferrule{
  }{
    \natty \le \natty
  }

  \inferrule{
  }{
    \unitty \le \unitty
  }

  \inferrule{
  }{
    \emptyty \le \emptyty
  }

  \inferrule{
    A' \le A \\
    \C \le \C'
  }{
    A \to \C \le A' \to \C'
  }

  \inferrule{
    \Rgn \subseteq \Rgn'
  }{
    E^\Rgn \le E^{\Rgn'}
  }

  \inferrule{
    \C' \le \C \\
    \D \le \D'
  }{
    \C \hto \D \le \C' \hto \D'
  }

  \inferrule{
    A \le A' \\
    \Drt \subseteq \Drt'
  }{
    A \E \Drt \le A' \E \Drt'
  }
\end{mathpar}
}%
It is easily checked that reflexivity and transitivity of subtyping are admissible.
Apart from resolving the poisoning problem, subtyping allows us to better deduce the behavior of handlers. Consider the computation
\begin{align*}
  &\letin{u = \val \iota} \\
  &\letin{v = (\conditional{b}{\val u}{\val \iota'})} {} \\
  &\letin{h = \val (\handler \val x \mapsto \cdots \effcase \call{u}{\op}{x}{k} \mapsto c)} \\
  &\quad \cdots
\end{align*}
Without subtyping we are forced to give both $u$ and $v$ the type $E^{\set{\iota,
    \iota'}}$. Therefore, by looking at the type of $u$ we cannot tell whether $h$ handles
$\hash{\iota}{\op}$ or $\hash{\iota'}{\op}$, and so we must assume that both may be
unhandled by $h$. With subtyping we may give $u$ the type $E^{\set{\iota}}$
which makes it clear that $h$ handles $\hash{\iota}{\op}$.

\subsection{Typing rules}
\label{sub:typing-rules}

There are two typing judgments,
\begin{equation*}
  \ctx \ent e \T A
  \qquad\text{and}\qquad
  \ctx \ent c \T \C
\end{equation*}
stating that an expression $e$ has a pure type $A$ and a computation $c$ has a dirty type
$\C$, respectively. Here $\ctx$ is a typing context of the form $x_1 : A_1, \ldots, x_n : A_n$. There is also a third, auxiliary typing judgment
\begin{equation*}
  \ctx \ent \ocs \T \C / \Drt  
\end{equation*}
which states that operation cases~$\ocs$ all have the same outgoing type $\C$ and are
guaranteed to handle operations in $\Drt$. 
Most of the typing rules in Figure~\ref{fig:typing} are standard, except for:
\begin{description}
\item[\rulename{Inst}]
  we check that $\iota$ is one of instances in~$\Rgn$, which in turn must be contained in $\iotas_E$.
\item[\rulename{Hand}]
  to check that a handler has type $A \E \Drt \hto B \E \Drt'$, we verify that it
  converts a computation of type $A \E \Drt$ to one of type $B \E \Drt'$, which involves
  checking three premises.
  First, the value case must take a value of type $A$ to a computation of type $B \E \Drt'$.
  Second, all operation cases must have outgoing type $B \E \Drt'$.
  Third, every operation in $\Drt$ is either guaranteed to be handled by $\ocs$, or is contained in $\Drt'$.
\item[\rulename{OpCases-Nil}, \rulename{OpCases-Cons}]
  the auxiliary typing judgment verifies that the operation cases have the given outgoing type, and that
  they cover the given dirt. The empty list $\ocsnil$ does not cover anything and has any outgoing type.
  The rule \rulename{OpCases-Cons} checks the first operation case, checks the others inductively, and verifies 
  that $\Drt \subseteq \Drt' \cupdot \hash{R}{\op}$, where
  \[
    \Drt' \cupdot \hash{\Rgn}{\op} \defeq
      \begin{cases}
        \Drt' \cup \set{\hash{\iota}{\op}} &\text{if $\Rgn = \set{\iota}$,}\\
        \Drt' & \text{otherwise.}
      \end{cases}
  \]
  The idea is that we can be sure that an operation case handles $\hash{\iota}{\op}$ only
  when the type of its instance is of the form $E^{\set{\hash{\iota}{\op}}}$.
\item[\rulename{Op}]
  we first check that $e$ and $\op$ belong to the same effect.
  Then we check that $\Drt$ covers not just all possible operations that the operation call may cause
  (recall that $\Rgn$ may contain more than one instance), but also any operations in the continuation~$c$.
  We may assume that $c$ has the same dirt, as we can use \rulename{SubComp} otherwise.
  We use the same reasoning in rules \rulename{IfThenElse}, \rulename{Match} and \rulename{Let}.
\item[\rulename{With}]
  handlers behave like functions from computations to computations.
\item[\rulename{SubExpr}, \rulename{SubComp}]
  these \emph{subsumption rules} allow us to always assign a bigger type.
\end{description}

\begin{figure}
\hrulefill
  \small
  \begin{mathpar}
  \inferrule[Var]{
    (x \T A) \in \ctx
  }{
    \ctx \ent x \T A
  }

  \inferrule[True]{
  }{
    \ctx \ent \tru \T \boolty
  }

  \inferrule[False]{
  }{
    \ctx \ent \fls \T \boolty
  }

  \inferrule[Zero]{
  }{
    \ctx \ent 0 \T \natty
  }

  \inferrule[Succ]{
    \ctx \ent e \T \natty
  }{
    \ctx \ent \succ e \T \natty
  }

  \inferrule[Unit]{
  }{
    \ctx \ent \unt \T \unitty
  }

  \inferrule[Fun]{
    \ctx, x \T A \ent c \T \C
  }{
    \ctx \ent \fun{x \T A} c \T A \to \C
  }

  \inferrule[Inst]{
    \iota \in \Rgn \subseteq \iotas_E
  }{
    \ctx \ent \iota \T E^\Rgn
  }

  \inferrule[Hand]{
    \ctx, x \T A \ent c_v \T B \E \Drt' \\
    \ctx \ent \ocs \T B \E \Drt' / \Drt'' \\
    \Drt \subseteq \Drt'' \cup \Drt'
  }{
    \ctx \ent (
      \handler
      \val x \T A \mapsto c_v \effcase \ocs
    ) \T A \E \Drt \hto B \E \Drt'
  }

  \inferrule[SubExpr]{
    \ctx \ent e \T A \\
    A \le A'
  }{
    \ctx \ent e \T A'
  }

  \medskip

  \inferrule[OpCases-Nil]{
  }{
    \ctx \ent \ocsnil_{\C} \T \C / \purely
  }

  \inferrule[OpCases-Cons]{
    \ctx \ent e \T E^\Rgn \\
    (\op \T A^\op \to B^\op) \in \sig_E \\\\
    \ctx, x \T A^\op, k \T B^\op \to \C \ent c \T \C \\
    \ctx \ent \ocs \T \C / \Drt' \\
    \Drt \subseteq \Drt' \cupdot \hash{R}{\op}
  }{
    \ctx \ent (\call{e}{\op}{x}{k} \mapsto c \effcase \ocs) \T \C / \Drt
  }

  \medskip

  \inferrule[IfThenElse]{
    \ctx \ent e \T \boolty \\
    \ctx \ent c_1 \T \C \\
    \ctx \ent c_2 \T \C
  }{
    \ctx \ent \conditional{e}{c_1}{c_2} \T \C
  }

  \inferrule[Match]{
    \ctx \ent e \T \natty \\
    \ctx \ent c_1 \T \C \\
    \ctx, x \T \natty \ent c_2 \T \C
  }{
    \ctx \ent \matchnat{e}{c_1}{x}{c_2} \T \C
  }

  \inferrule[Absurd]{
    \ctx \ent e \T \emptyty
  }{
    \ctx \ent \kord{absurd}_{\C} \; e \T \C
  }
  
  \inferrule[App]{
    \ctx \ent e_1 \T A \to \C \\
    \ctx \ent e_2 \T A
  }{
    \ctx \ent e_1 \, e_2 \T \C
  }

  \inferrule[Val]{
    \ctx \ent e \T A
  }{
    \ctx \ent \val e \T A \E \Drt
  }

  \inferrule[Op]{
    \ctx \ent e_1 \T E^\Rgn \\
    (\op \T A^\op \to B^\op) \in \sig_E \\\\
    \ctx \ent e_2 \T A^\op \\
    \ctx, y \T B^\op \ent c \T A \E \Drt \\
    \fra{\iota \in \Rgn} \hash{\iota}{\op} \in \Drt
  }{
    \ctx \ent \call{e_1}{\op}{e_2}{\cont{y}{c}} \T A \E \Drt
  }

  \inferrule[Let]{
    \ctx \ent c_1 \T A \E \Drt \\
    \ctx, x \T A \ent c_2 \T B \E \Drt
  }{
    \ctx \ent \letin{x = c_1} c_2 \T B \E \Drt
  }

 \inferrule[LetRec]{
    \ctx, f \T A \to \C, x \T A \ent c_1 \T \C \\
    \ctx, f \T A \to \C \ent c_2 \T \D
  }{
    \ctx \ent \letrecin{f \, x \T A \to \C = c_1} c_2 \T \D
  }

  \inferrule[With]{
    \ctx \ent e \T \C \hto \D \\
    \ctx \ent c \T \C
  }{
    \ctx \ent \withhandle{e}{c} \T \D
  }

  \inferrule[SubComp]{
    \ctx \ent c \T \C \\
    \C \le \C'
  }{
    \ctx \ent c \T \C'
  }
\end{mathpar}
\caption{The typing rules of core \Eff.}
\label{fig:typing}
\hrulefill
\end{figure}

The effect system is safe with respect to the operational semantics:

\begin{thm}[Progress \& Preservation]
  \hfill
  \begin{description}
    \item[Progress]
      If $\ent c \T A \E \Drt$ then either
      \begin{itemize}
        \item there exists a computation $c'$ such that $c \step c'$, or
        \item $c$ is of the form $\val e$ for some expression $e$, or
        \item $c$ is of the form $\call{\iota}{\op}{e}{\cont{x}{c'}}$ for some $\hash{\iota}{\op} \in \Drt$.
      \end{itemize}
    \item[Preservation]
      If $\ent c \T \C$ and $c \step c'$ then $\ent c' \T \C$.
  \end{description}
\end{thm}

\begin{proof}
  Both statements are proved by induction.
  The formalized proofs can be found in the files \verb|progress.elf| and \verb|preservation.elf|.
\end{proof}

\begin{cor}[Safety]
  A terminating computation of type $A \E \Drt$ returns a value of type~$A$, or calls an operation in~$\Drt$.
\end{cor}

\noindent
In particular, a terminating computation of type $A \E \purely$ does not call any operations and returns
a pure value of type $A$.

\begin{exa}
Let us see what the effect system tells us about Example~\ref{exa:reduction}.
We may give the reference $\iota \in \iotas_{\kord{ref}}$ type $\kord{ref}^{\set{\iota}}$.
Then the computation~$c$ has the dirty type
$\natty \E \set{\hash{\iota}{\kord{lookup}}, \hash{\iota}{\kord{update}}}$,
and the handler~$h$ the type
\[
  (\natty \E \set{\hash{\iota}{\kord{lookup}}, \hash{\iota}{\kord{update}}}) \hto (\unitty \E \set{\hash{\iota}{\kord{update}}})
\]
because it handles both $\kord{lookup}$ and $\kord{update}$, but then calls $\kord{update}$ in the value case.
This $\kord{update}$ also changes the type of computation from $\natty$ to $\unitty$.

If we give $\iota$ the less precise type $\kord{ref}^{\set{\iota, \iota'}}$,
the dirt of $c$ is 
\[
  \Drt \defeq \set{\hash{\iota}{\kord{lookup}}, \hash{\iota'}{\kord{lookup}}, \hash{\iota}{\kord{update}}, \hash{\iota'}{\kord{update}}}
\]
while the best type we can give to $h$ is $(\natty \E \Drt) \hto (\unitty \E \Drt)$.
Since $\set{\iota, \iota'}$ is not a singleton,
we cannot give any guarantees on what operations are handled.
\end{exa}

\subsection{Skeletal types}
\label{sub:skeletal-types}

We relate the pure and dirty types to ML-style types by an operation which erases all effect information to produce a \emph{skeletal} type. We will use these later to obtain a coherent semantics of types. The skeletal types are defined as follows:
\begin{align*}
  \text{Skeletal type}\ S, T &\bnfis
  \boolty \bnfor
  \natty \bnfor
  \unitty \bnfor
  \emptyty \bnfor
  S \to T \bnfor
  E \bnfor
  S \hto T
\end{align*}
There is no distinction between pure and dirty types anymore. The typing rules for skeletal types are like those for pure and dirty types with the effect information omitted, for instance
{\small
  \begin{mathpar}
    \inferrule[Inst']{
      \iota \in \iotas_E
    }{
      \ctx \ent \iota \T E
    }

    \inferrule[Hand']{
      \ctx, x \T A^s \ent c_v \T S \\
      \ctx \ent \ocs \T S
    }{
      \ctx \ent (
      \handler
      \val x \T A \mapsto c_v \effcase \ocs
      ) \T A^s \hto S
    }

    \inferrule[OpCases-Nil']{
    }{
      \ctx \ent \ocsnil_{\C} \T \C^s
    }
    
    \inferrule[OpCases-Cons']{
      \ctx \ent e \T E \\
      (\op \T A^\op \to B^\op) \in \sig_E \\
      \ctx, x \T (A^\op)^s, k \T (B^\op)^s \to S \ent c \T S \\
      \ctx \ent \ocs \T S
    }{
      \ctx \ent (\call{e}{\op}{x}{k} \mapsto c \effcase \ocs) \T S
    }

    \inferrule[Op']{
      \ctx \ent e_1 \T E \\
      (\op \T A^\op \to B^\op) \in \sig_E \\
      \ctx \ent e_2 \T (A^\op)^s \\
      \ctx, y \T (B^\op)^s \ent c \T S
    }{
      \ctx \ent \call{e_1}{\op}{e_2}{\cont{y}{c}} \T S
    }

  \end{mathpar}
}%
The remaining rules remain unchanged as they do not mention effects,
while subsumption rules are \emph{removed}. To every pure type $A$ and a dirty type $\C$ we assign their skeletal versions $A^s$ and $\C^s$, which are like $A$ and $\C$ with region and dirt removed. The skeletal version $\ctx^s$ of a typing context $\ctx$ is obtained by taking the skeletons of the types in $\ctx$.
We summarize the properties of skeletal types:

\begin{thm}
  \label{thm:skeletal}
  \parbox{0pt}{}
  \begin{enumerate}
  \item If $A \leq B$ and $\C \leq \D$ then $A^s = B^s$ and $\C^s = \D^s$.
  \item If
    \begin{equation*}
      \ctx \ent e : A
      \qquad\text{and}\qquad
      \ctx \ent c : \C
    \end{equation*}
    then
    \begin{equation*}
      \ctx^s \ent e : A^s
      \qquad\text{and}\qquad
      \ctx^s \ent c : \C^s.
    \end{equation*}
  \item In a given context, an expression and a computation has at most one skeletal type, with a unique typing derivation.
  \end{enumerate}
\end{thm}

\begin{proof}
  The first statement holds by an induction on the derivation of $A \leq B$ and $\C \leq \D$.

  To prove the second statement, observe that
  a typing derivation may be mapped to the corresponding skeletal version rule by rule, except for subsumption rules. But these can be simply omitted from the typing derivations because $A \leq B$ and $\C \leq \D$ imply $A^s = B^s$ and $\C^s = \D^s$ by the first statement.

  The last statement holds by inversion: in any situation at most one skeletal typing rule applies in at most one way.
\end{proof}

A consequence of the theorem is that if a computation $c$ has dirty types $\C$ and $\D$ then $\C^s = \D^s$, hence $\C$ and $\D$ differ only in the effect information. An analogous property holds for expressions and pure types.

%%% Local Variables: 
%%% mode: latex
%%% TeX-master: "effect-system"
%%% End: 

\section{Denotational semantics}
\label{sec:denotational-semantics}

We use standard domain theory to provide an adequate denotational semantics of core \Eff. We shall use $\omega$-cpos as domains, but presumably a different kind of domains could be used, as long as they support the standard constructions, in particular solutions of domain equations, and are amenable to Pitts's theory of minimal invariant properties~\cite{Pitts96}. We refer to~\cite{Amadio:1998} for background on domain theory and denotational semantics.

We define a \emph{predomain} to be a poset in which chains (ascending sequences) have suprema, while a \emph{domain} is a predomain with a least element~$\bot$. A \emph{continuous} map is a monotone map which commutes with suprema of chains. If $D$ is a predomain and $E$ is a domain the set $D \to E$ of all continuous maps forms a domain. The ordering on continuous maps is pointwise. A continuous map is \emph{strict} if it maps $\bot$ to $\bot$. The set $D \stricto E$ of strict maps between domains $D$ and $E$ forms a subdomain of $D \to E$.

\subsection{Computation domains}
\label{sub:computation-domains}

We first build domains that will serve as the meanings of computation types. Let $A$ be a predomain, $I$ an index set, and for each $i \in I$ let $A_i$ and $B_i$ be predomains. We seek a domain $\compdom$ satisfying the domain equation $\compdom \cong F(\compdom)$ where $F$ is the functor
\begin{equation*}
  \textstyle
  F(D) = \big(A + \coprod_{i \in I} A_i \times ({B_i} \to D)\big)_\bot.
\end{equation*}
Following~\cite{Pitts96} we work in the category of domains and strict maps, and take $\compdom$ to be a minimal solution in the sense that it possesses the \emph{minimal invariant property}. The usual limit-colimit construction~\cite[Chapter~7]{Amadio:1998} yields such a domain.  As domain equations go, this one is quite simple because $\compdom$ occurs only covariantly. The elements of $\compdom$ can be thought of as trees whose leaves are tagged with elements of~$A$ or~$\bot$, and whose nodes have branching types $B_i$ and are tagged with elements of~$A_i$. The trees need not be well founded.

The minimality of~$\compdom$ yields a recursion and an induction principle. The recursion principle says that for any domain $D$, a continuous map $f_\kord{val} : A \to D$, and continuous maps $f_i : A_i \times ({B_i} \to D) \to D$ for $i \in I$, there is a unique strict continuous map $f : \compdom \stricto D$ such that
\begin{align*}
  f(\inval(x)) &= f_{\kord{val}}(x) & &\text{for $x \in A$,}\\
  f(\inop{i}(y, \kappa)) &= f_i (y, f \circ \kappa)  & &\text{for $y \in A_i$ and $\kappa : B_i \to \compdom$.}
\end{align*}
The induction principle for $\compdom$ applies to \emph{admissible} predicates on $\compdom$, i.e., those that hold for $\bot$ and are closed under suprema of chains. Precisely, if $\phi$ is an admissible predicate on $\compdom$ such that
\begin{enumerate}
\item $\phi(\inval(x))$ for all $x \in A$, and
\item for all $i \in I$, $x \in A_i$, and $\kappa : B_i \to \compdom$, if $\fra{y \in B_i} \phi(\kappa(y))$ then $\phi(\inop{i}(x,\kappa))$,
\end{enumerate}
then $\phi(t)$ holds for all $t \in \compdom$.

For any $I$, the construction of a minimal solution $\compdom_I(A, (A_i)_i, (B_i)_i)$
from the input data $A$, $(A_i)_{i \in I}$, $(B_i)_{i \in I}$ forms a locally continuous functor
\begin{equation*}
  T_I :
  \mathbf{pCpo} \times \mathbf{pCpo}^I \times (\mathbf{pCpo}^\mathrm{op})^I
  \to
  \mathbf{Cppo}.
\end{equation*}
Here $\mathbf{Cppo}$ is the category of domains and strict continuous maps, while
$\mathbf{pCpo}$ is the category of predomains and partial continuous maps which are
defined on open subsets (an upper set which is inaccessible by suprema of chains). The
functor will appear later on in a larger system of recursive domain equations.

\subsection{Semantics of skeletal types}
\label{sub:semantics-skeletal-types}

We interpret pure and dirty types as predomains and domains, respectively. A typing context is interpreted as a cartesian product of predomains, and a typing judgment as a continuous map. However, typing judgments do not have unique derivations because of the subsumption rules, and so we have to worry about coherence. That is, when we define the meaning of a typing judgment by induction on its derivation, we need to make sure that the result does not depend on the choice of derivation.
We accomplish this by providing a semantics which factors through the skeletal types from Section~\ref{sub:skeletal-types}.

Let
\begin{equation*}
  \Omega = \set{\hash{\iota}{\op} \mid \exs{E} \iota \in \iotas_E \land \op \in \sig_E}
\end{equation*}
be the set of \emph{all} operations. If dirts could be infinite, $A \E \Omega$ would be a
dirty type expressing the fact that any effect could happen.

To each skeletal type $S$ we assign a predomain $\skele{S}$ and a domain $\skelc{S}$ as follows:
\begin{align*}
  \skele{\boolty} &= \set{\semfls, \semtru}, &
  \skele{\natty} &= \NN, \\
  \skele{\unitty} &= \set{\star}, &
  \skele{\emptyty} &= \purely, \\
  \skele{S \to T} &= \skele{S} \to {\skelc{T}}, &
  \skele{E} &= \iotas_E, \\
  \skele{S \hto T} &= \skelc{S} \stricto \skelc{T}, &
\end{align*}
and
\begin{equation*}
  \skelc{S} = \compdom_\Omega(\skele{S},
                       (\skele{(A^{\op})^s})_{\hash{\iota}{\op}}, 
                       (\skele{(B^{\op})^s})_{\hash{\iota}{\op}}).
\end{equation*}
These should be read as a system of domain and predomain equations indexed by the skeletal types. There are possible circularities in the system because the equation for $\skelc{S}$ refers to possibly larger skeletal types $(A^\op)^s$ and $(B^\op)^s$. As in the case of computation domains, we take the minimal solutions which enjoy the minimal invariant property.

To each pure type $A$ and dirty type $\C$ we assign a \emph{skeletal predomain} $\skel{A}$ and \emph{skeletal domain} $\skel{\C}$ by setting
\begin{equation*}
  \skel{A} = \skele{A^s}
  \qquad\text{and}\qquad
  \skel{\C} = \skelc{\C^s}.
\end{equation*}
The first part of Theorem~\ref{thm:skeletal} guarantees that $A \leq B$ and $\C \leq \D$ imply $\skel{A} = \skel{B}$ and $\skel{\C} = \skel{\D}$.

\subsection{Semantics of expressions and computations}
\label{sub:semant-expr-comp}

The meaning of a typing context $\ctx$
\begin{equation*}
  x_1 : A_1, \ldots, x_n : A_n
\end{equation*}
is
\begin{equation*}
  \skel{\ctx} =
  \skel{A_1} \times \cdots \times \skel{A_n}.
\end{equation*}
We interpret typing judgments
\begin{align*}
  \ctx \ent e : A
  \qquad\text{and}\qquad
  \ctx \ent c : \C
\end{align*}
as continuous maps
\begin{equation*}
  \xsem{\ctx \ent e : A} : \skel{\ctx} \to \skel{A}
  \qquad\text{and}\qquad
  \xsem{\ctx \ent c : \C} : \skel{\ctx} \to \skel{\C}.
\end{equation*}
When no confusion can arise we abbreviate these as $\xsem{e}$ and $\xsem{c}$.
The definition proceeds by induction on the derivation of the typing judgment.
Given an environment $\eta \in \skel{\Gamma}$, the base rules for expressions and the
successor rule are taken care of by
\begin{align*}
  \sem{\ctx \ent x_i : A_i}{\eta} &= \eta_i &
  \sem{\ctx \ent \fls : \boolty}{\eta} &= \semfls \\
  \sem{\ctx \ent \unt : \unitty}{\eta} &= \star  &
  \sem{\ctx \ent \tru : \boolty}{\eta} &= \semtru \\
  \sem{\ctx \ent 0 : \natty}{\eta} &= 0 &
  \sem{\ctx \ent \succ e : \natty}{\eta} &= (\sem{\ctx \ent e : \natty}{\eta}) + 1 \\
  \sem{\ctx \ent \iota : E^R}{\eta} &= \iota,
\end{align*}
and the abstraction rule by
\begin{equation*}
    \sem{\ctx \ent (\fun{x \T A} c) : A \to \C}{\eta} =
    \lambda a \in \skel{A} \,.\, \sem{\ctx, x : A \ent c : \C}{(\eta, a)}
\end{equation*}
For the handler rule we set
\begin{equation*}
  \sem{\ctx \ent (\handler \val x \T A \mapsto c_v \effcase \ocs) : A \E \Drt \hto B \E \Drt'} = h
\end{equation*}
where $h : \skel{\ctx} \to \skel{A \E \Drt} \stricto \skel{B \E \Drt'}$ is defined by recursion on $\skel{A \E \Drt}$:
\begin{align*}
  h(\eta)(\bot) &= \bot \\
  h(\eta)(\inval(a)) &= \sem{\ctx, x : A \ent c_v : B \E \Drt'}{(\eta, a)} \\
  h(\eta)(\inop{\hash{\iota}{\op}}(a, \kappa)) &=
  \xsem{\ocs}_{\hash{\iota}{\op}}(\eta, a, h(\eta) \circ \kappa)
\end{align*}
The auxiliary map
\begin{equation*}
  \xsem{\ocs}_{\hash{\iota}{\op}} :
  \skel{\Gamma} \times \skel{A^{\op}} \times \skel{B^{\op} \to B \E \Drt'} \to
  \skel{B \E \Drt'}
\end{equation*}
is defined by
\begin{multline*}
  \begin{aligned}
    \xsem{\ocsnil_{\C}}_{\hash{\iota}{\op}}(\eta, a, \kappa) &= \inop{\hash{\iota}{\op}}(a, \kappa) \\
    \xsem{\call{e'}{\op'}{x}{k} \mapsto c \effcase \ocs}_{\hash{\iota}{\op}}(\eta, a, \kappa) &=    
  \end{aligned} \\
  \begin{cases}
    \sem{\ctx, x : A^{\op}, k : B^{\op} \to B \E \Drt' \ent c : B \E \Drt'}{(\eta, a, \kappa)}
    &\text{if $\hash{(\sem{e'}{\eta})}{\op'} = \hash{\iota}{\op}$,}\\
    \xsem{\ocs}_{\hash{\iota}{\op}}(\eta, a, \kappa) &\text{otherwise.}
  \end{cases}
\end{multline*}
Finally, if $\ctx \ent e : A'$ is derived from the premises $\ctx \ent e: A$ and $A \leq A'$ by the subsumption rule, we set
\begin{equation*}
  \sem{\ctx \ent e : A'}{\eta} = \sem{\ctx \ent e : A}{\eta}.
\end{equation*}
The definition is meaningful because $A \leq A'$ implies $\skel{A} = \skel{A'}$.
The meaning of pure computations and operations is
\begin{align*}
  \sem{\ctx \ent \val e : A \E \Drt}{\eta} &= \inval(\sem{\ctx \ent e : A}) \\
  \sem{\ctx \ent \call{e_1}{\op}{e_2}{\cont{y}{c}} : A \E \Drt}{\eta} &=
  \inop{\hash{\sem{e_1}{\eta}}{\op}} (\sem{e_2}{\eta},
        \lambda b \in \skel{B^{\op}} \,.\, \sem{\ctx, y : B^{\op} \ent c : A \E \Drt}{(\eta,b)}).
\end{align*}
The meaning of elimination forms is
\begin{align*}
  \sem{\ctx \ent \withhandle{e}{c} : \D}{\eta} &=
     (\sem{\ctx \ent e : \C \hto \D}{\eta}) (\sem{\ctx \ent c : \C}{\eta}) \\
  \sem{\ctx \ent \conditional{e}{c_1}{c_2} : \C}{\eta} &=
  \begin{cases}
    \sem{\ctx \ent c_1 : \C}{\eta} & \text{if $\sem{\ctx \ent e : \boolty}{\eta} = \semtru$,}\\
    \sem{\ctx \ent c_2 : \C}{\eta} & \text{if $\sem{\ctx \ent e : \boolty}{\eta} = \semfls$}
  \end{cases}
  \\
  \sem{\ctx \ent \absurd e : \C}{\eta} &= \bot \\
  \sem{\ctx \ent e_1 \, e_2 : \C}{\eta} &= (\sem{\ctx \ent e_1 : A \to \C}{\eta}) (\sem{\ctx \ent e_2 : A})
\end{align*}
and
\begin{multline*}
    \sem{\ctx \ent (\matchnat{e}{c_1}{x}{c_2}) : \C}{\eta} = {} \\
  \begin{cases}
    \sem{\ctx \ent c_1 : \C}{\eta} & \text{if $\sem{\ctx \ent e : \natty}{\eta} = 0$,}\\
    \sem{\ctx, x : \natty \ent c_2 : \C}{(\eta,n)} & \text{if $\sem{\ctx \ent e: \natty}{\eta} = n + 1$.}
  \end{cases}
\end{multline*}
To give semantics of $\kord{let}$ binding, we first define the \emph{lifting} of a map $f
: \skel{A} \to \skel{B \E \Drt}$ to be the map $\lift{f} : \skel{A \E \Drt} \to \skel{B \E
  \Drt}$ defined recursively by
\begin{align*}
  \lift{f}(\bot) &= \bot, \\
  \lift{f}(\inval(x)) &= f(x), \\
  \lift{f}(\inop{\hash{\iota}{\op}}(x, \kappa)) &= \inop{\hash{\iota}{\op}}(x, \lift{f} \circ \kappa).
\end{align*}
Then we set
\begin{multline*}
  \sem{\ctx \ent \letin{x = c_1} c_2 : B \E \Drt}{\eta} = \\
  \lift{(\lambda a \in \skel{A} \,.\, \sem{\ctx, x : A \ent c_2 : B \E \Drt}{(\eta, a)})}
  (\sem{\ctx \ent c_1 : A \E \Drt}{\eta}).
\end{multline*}
The meaning of a recursive function definition is
\begin{equation*}
  \sem{\ctx \ent (\letrecin{f \, x \T A \to \C = c_1} c_2) : \C}{\eta} =
  \sem{\ctx, f : A \to \C \ent c_2 : \C}{(\eta, g)}
\end{equation*}
where $g : \skel{A} \to \skel{\C}$ is the least fixed-point of the map
\begin{equation*}
  g \mapsto (\lambda a \in \skel{A} \,.\, \sem{c_1}{(\eta, g, a)}).
\end{equation*}
Finally, just like for expressions, if $\ctx \ent c : \C'$ is derived from the premises $\ctx \ent c : \C$ and $\C \leq \C'$ by the subsumption rule, we set
\begin{equation*}
  \sem{\ctx \ent c : \C'}{\eta} = \sem{\ctx \ent c : \C}{\eta}.
\end{equation*}
This concludes the definition of denotational semantics of expressions and computations.

\begin{thm}[Coherence]
  All derivations of a typing judgment give it the same meaning.
\end{thm}

\begin{proof}
  The semantics of $\ctx \ent e : A$ and $\ctx \ent c : \C$ factor through the associated
  skeletal derivations $\ctx^s \ent e : A^s$ and $\ctx^s \ent c : \C^s$, respectively. This
  is so because the semantic rules given above clearly factor through the associated
  skeletal rules. In other words, they ignore the effect information and the subsumption
  rules. Uniqueness of meaning is thus a consequence of the uniqueness of skeletal
  derivations, cf.\ Theorem~\ref{thm:skeletal}.
\end{proof}

\subsection{Semantics of effects}

In the terminology of John Reynolds~\cite{reynolds00themeaning} the semantics given so far is \emph{intrinsic}, while the semantics of effects given below is \emph{extrinsic}. 
This is in accordance with our understanding that effect information is descriptive rather than prescriptive: a function $f$ of type $(A \to B \E \Drt) \to B' \E \Drt'$ should accept \emph{any} function $g$ of type $A \to B \E \Drt''$, even if $\Drt'' \not\subseteq \Drt$, although it has the \emph{property} that if $\Drt'' \subseteq \Drt$ then $f(g)$ calls only operations in $\Drt'$.

For each pure type $A$ and dirty type $\C$ we define subpredomains $\tsem{A} \subseteq \skel{A}$ and subdomains $\tsem{\C} \subseteq \skel{\C}$ of those elements that behave according to the effect information. The ground types are easy:
\begin{align*}
  \tsem{\boolty} &= \skel{\boolty}, &
  \tsem{\natty} &= \skel{\natty}, &
  \tsem{\unitty} &= \skel{\unitty}, \\
  \tsem{\emptyty} &= \skel{\emptyty}, &
  \tsem{E^R} &= R.
\end{align*}
For function types, handler types, and dirty types we would like to solve the following system of equations with unknowns $\tsem{A}$ and $\tsem{\C}$ where $A$ and $\C$ range over pure and dirty types, respectively:
\begin{align}
  \notag
  \tsem{A \to \C} &= \set{f \in \skel{A} \to \skel{\C} \such \fra{x \in \tsem{A}} f(x) \in \tsem{\C}}, \\
  \notag
  \tsem{\C \hto \D} &= \set{h \in \skel{\C} \stricto \skel{\D} \such \fra{t \in \tsem{\C}} h(t) \in \tsem{\D}}, \\
  \label{eq:tsem3}
  \tsem{A \E \Drt} &= \{t \in \skel{A \E \Drt} \such
    \begin{aligned}[t]
      & t = \bot \lor (\exs{x \in \tsem{A}} t = \inval(x)) \lor {}\\
      & \exs{\hash{\iota}{\op} \in \Drt, y \in \tsem{A^\op}, \kappa \in \skel{B^\op \to A \E \Drt}} {} \\
      & {}\quad (\fra{z \in \tsem{B^\op}} \kappa(b) \in \tsem{A \E \Drt}) \land
        t = \inop{\hash{\iota}{\op}}(y, \kappa)
      \}.
    \end{aligned}
\end{align}
The last rule says that $t \in \tsem{A \E \Drt}$ when it is $\bot$, or of the form $\inval(x)$ for some $x \in \tsem{A}$, or of the form $\inop{\hash{\iota}{\op}}(y, \kappa)$ for some $y \in \tsem{A^\op}$ and $\kappa \in \tsem{B^\op \to A \E \Drt}$.

The system is potentially problematic because the types $A^\op$ and $B^\op$ introduce circularities in the last equation. We apply Pitts's theorem~\cite[Theorem 4.16]{Pitts96} about existence of invariant relations to obtain a solution that satisfies an induction principle, see Theorem~\ref{thm:tsem-induction} below. For Pitts's theorem to apply we must verify that our conditions form an admissible action on admissible relations, as defined in~\cite[Definition 4.6]{Pitts96}. The relational structure in question is that of~\cite[Example~4.2(ii)]{Pitts96} for which the accompanying notion of admissibility is the one we are using, cf.~\cite[Example~4.5(ii)]{Pitts96}. The locally continuous functor is the evident one, while its admissible action on relations is read off~\eqref{eq:tsem3}. The admissibility of the action follows from \cite[Lemma~6.6]{Pitts96}: the first two actions in~\eqref{eq:tsem3} are considered explicitly, while the third one is a composition of actions from the cited lemma.
The solution so obtained possesses the following induction principle.

\begin{thm}
  \label{thm:tsem-induction}%
  Suppose $\phi$ is an admissible predicate on $\skel{A \E \Drt}$ such that
  \begin{enumerate}
  \item $\phi(\inval(a))$ for every $a \in \tsem{A}$, and
  \item for all $\hash{\iota}{\op} \in \Drt$, $a \in \tsem{A^\op}$, $\kappa \in \tsem{B^\op \to A \E \Drt}$,
    if $\fra{b \in \tsem{B^\op}} \phi(\kappa(b))$ then $\phi(\inop{\hash{\iota}{\op}}(a, \kappa))$.
  \end{enumerate}
  Then $\phi(t)$ for all $t \in \tsem{A \E \Drt}$.
\end{thm}

\begin{proof}
  Because we defined $\tsem{{-}}$ mutually for all types, we first extend $\phi$ to be
  constantly true on types other than $A \E \Drt$. Then our theorem becomes an instance of
  the induction property \cite[Theorem~6.5]{Pitts96}. The admissible action required for
  the application of the theorem is obtained using \cite[Lemma~6.6]{Pitts96}.
\end{proof}

The semantics of effects and the semantics of expressions and computations fit together:

\begin{thm}
  If $\ctx \ent e : A$ and $\ctx \ent c : \C$ then for all $\eta \in \tsem{\ctx} = \tsem{A_1} \times \cdots \times \tsem{A_n}$, we have
  \begin{equation*}
    \sem{\ctx \ent e : A}{\eta} \in \tsem{A}
    \qquad\text{and}\qquad
    \sem{\ctx \ent c : \C}{\eta} \in \tsem{C}.
  \end{equation*}
\end{thm}

\begin{proof}
  The proof proceeds by induction on the derivation of the typing judgment. All cases are easy, except for the typing rule
  for a handler
  \begin{equation*}
    \ctx \ent (\handler \val x \mapsto c_v \effcase \ocs) : A \E \Drt \hto B \E \Drt'
  \end{equation*}
  First, we claim that if the auxiliary judgment
  \begin{equation*}
    \ctx \ent \ocs \T B \E \Drt' / \Drt_c
  \end{equation*}
  is derivable then
  $\xsem{\ocs}_{\hash{\iota}{\op}}(\eta, a, \kappa) \in \tsem{B \E \Drt'}$
  for all $\hash{\iota}{\op} \in \Drt' \cup \Drt_c$, $\eta \in \tsem{\ctx}$, $a \in
  \tsem{A^\op}$, and $\kappa : \tsem{B^\op \to B \E \Drt'}$. Assuming the claim has been
  established, the above handler rule follows immediately.
  
  The proof of the claim proceeds by induction on the derivation of the auxiliary
  judgment. For $\ocsnil$, we have $\Drt_c = \purely$ and the claim obviously holds.
  The other possibility is
  \begin{equation*}
    \ctx \ent (\call{e}{\op'}{x}{k} \mapsto c \effcase \ocs) : B \E \Drt'/ \Drt_c'
  \end{equation*}
  for some $\ctx \ent \ocs : B \E \Drt' / \Drt_c$ and $\Drt_c' \subseteq (\Drt_c \cupdot \hash{R}{\op})$.
  Define $\iota' = \sem{\ctx \ent e : E^R}{\eta}$ and consider two cases. First, if
  $\hash{\iota'}{\op'} = \hash{\iota}{\op}$ then
  \begin{equation*}
    \xsem{\call{e}{\op'}{x}{k} \mapsto c \effcase \ocs}_{\hash{\iota}{\op}}(\eta, a, \kappa) =
    \xsem{\ctx, x : A^\op, k : B^\op \to B \E \Drt' \ent c : B \E \Drt'}(\eta, a, \kappa),
  \end{equation*}
  and we may apply the induction hypothesis (for the whole theorem) to~$c$. Second,
  suppose $\hash{\iota'}{\op'} \neq \hash{\iota}{\op}$. Then the assumption
  $\hash{\iota}{\op} \in \Drt' \cup \Drt_c'$ implies
  $\hash{\iota}{\op} \in \Drt' \cup \Drt_c$. Indeed, if $\hash{\iota}{\op} \in \Drt'$ there
  is nothing to prove, and if $\hash{\iota}{\op} \in \Drt_c' \subseteq (\Drt_c \cupdot \hash{R}{\op})$,
  then either $R$ is not a singleton and $\Drt_c \cupdot \hash{R}{\op} = \Drt_c$, or $R =
  \set{\hash{\iota'}{\op'}}$ and $\hash{\iota}{\op} \not\in R$. In any case, it follows
  that
  \begin{equation*}
    \xsem{\call{e}{\op'}{x}{k} \mapsto c \effcase \ocs}_{\hash{\iota}{\op}}(\eta, a, \kappa) =
    \xsem{\ocs}_{\hash{\iota}{\op}}(\eta, a, \kappa)
  \end{equation*}
  and $\ctx \ent \ocs : B \E \Drt'/\Drt_c$. Thus we may apply the induction hypothesis.  
\end{proof}

\subsection{Soundness and adequacy}
\label{sub:soundness-adequacy}

The soundness and adequacy theorems state that operational semantics and denotational
semantics fit together. One is an easy induction, while the other is proved using standard
technique of formal approximation relations, e.g., see~\cite[Theorem~6.3.6]{Amadio:1998}.

\begin{thm}[Soundness]
  If $\ent c : \C$ and $c \step c'$ then $\xsem{\ent c : \C} = \xsem{\ent c' : \C}$.
\end{thm}

\begin{proof}
  We just have to walk through all the defining rules for $\step$ and verify that they do
  indeed preserve the meaning of $c$.
\end{proof}

To tackle adequacy, see Corollary~\ref{cor:adequacy} below, we define formal approximation
relations $\aprx{A}$ and $\aprx{\C}$ whose intuitive meaning is that an element of
$\skel{A}$ or $\skel{\C}$ approximates a closed term of type $A$ or $\C$. Given $d \in
\skel{A}$ and a closed expression $e : A$, we define $d \aprx{A} e$ by:
\begin{align*}
  d \aprx{\boolty} e &\iff (d = \semtru \land e = \tru) \lor (d = \semfls \land e = \fls) \\
  d \aprx{\natty} e &\iff d = n \land e = \kord{succ}^n \, 0 \\
  d \aprx{\unitty} e &\iff d = \star \land e = \unt \\
  d \aprx{A \to \C} e &\iff \fra{d', e'} (d' \aprx{A} e' \Rightarrow d(d') \aprx{\C} e \, e') \\
  d \aprx{E^R} e &\iff d = e \\
  d \aprx{\C \hto \D} e &\iff \fra{d', c} (d' \aprx{\C} c \Rightarrow d(d') \aprx{\D} (\withhandle{e}{c}))
\end{align*}
Simultaneously we define when $d \in \skel{A \E \Drt}$ approximates a closed computation $c
: A \E \Drt$, where $d \aprx{A \E \Drt} c$ holds when
\begin{itemize}
\item $d = \bot$, or
\item $d = \inval(d')$, $c \eval \val e$ and $d' \aprx{A} e$, or
\item $d = \inop{\hash{\iota}{\op}}(d', \kappa)$, $c \eval \call{\iota}{\op}{e}{\cont{y}{c'}}$, $d' \aprx{A^\op} e$,
  and if $d'' \aprx{B^\op} e'$ then $\kappa(d'') \aprx{A \E \Drt} c'[e''/y]$.
\end{itemize}
The definition of $\aprx{A \E \Drt}$ refers to types $A^\op$ and $B^\op$ which are
possibly larger than $A$, thus we again use \cite[Theorem~4.16]{Pitts96} to establish
existence of minimal such $\aprx{A}$ and $\aprx{\C}$, much like for the recursive types
considered in~\cite[Section~5]{Pitts96}.

\begin{lem}
  \label{lem:approx-step}%
  If $d \aprx{A \E \Drt} c'$ and $c \step c'$ then $d \aprx{A \E \Drt} c$.
\end{lem}

\begin{proof}
  There is nothing to prove if $d = \bot$. The other two cases
  follow because $c \step c'$ and $c' \eval r$ imply $c \eval r$. For instance, if $d
  \aprx{A \E \Drt} c'$ holds because $d = \inval(d')$, $c \eval \val e$ and $d' \aprx{A}
  e$ for some $d'$ and $e$, then $d \aprx{A \E \Drt} c$ holds because $c \step c'$ implies
  $c \eval \val e$ and so we may reuse $d'$ and $e$.
\end{proof}

\begin{lem}
  \label{lem:approx-sup}%
  The relations $\aprx{A}$ and $\aprx{\C}$ are closed under suprema of chains in the first argument.
  The relations $\aprx{\C}$ relate $\bot$ to every computation.
\end{lem}

\begin{proof}
  The second statement holds by the definition of $\aprx{\C}$.
  For the first statement we proceed by induction on the type. The base cases hold because the
  predomains for ground types are flat. For $A \to \C$, suppose $(d_n)_n$ is a chain in
  $\skel{A \to \C}$ and $d_n \aprx{A \to \C} e$ for all~$n$. Consider any $d', e'$ such
  that $d' \aprx{A} e'$. Then $d_n(d') \aprx{\C} e \, e'$ for all~$n$.
  Suprema in $\skel{A \to \C}$ are defined pointwise, so we can use the induction hypothesis for $\C$ and get
  \begin{equation*}
    \textstyle
    (\bigvee_n d_n)(d') = \bigvee_n d_n(d') \aprx{\C} e \, e',
  \end{equation*}
  hence $\bigvee_n d_n \aprx{\C} e$ as desired. Handler types are treated similarly.
  Consider a computation type $A \E \Drt$, a closed computation $c : A \E \Drt$ and a chain $(d_n)_n$ in $\skel{A \E \Drt}$
  such that $d_n \aprx{\C} c$ for all~$n$. There are three kinds of chains in $\skel{A \E
    \Drt}$:
  \begin{itemize}
  \item $(d_n)_n$ is constantly $\bot$: then $\bigvee_n d_n = \bot \aprx{\C} c$.
  \item for large enough $n$ there are $d_n'$ such that $d_n = \inval(d_n')$: then
    $(d_n')_n$ form a chain in $\skel{A}$. Because operational semantics is deterministic,
    there exists a single $e$ such that $c \eval \val e$, and $d_n' \aprx{A} e$. By
    induction hypothesis for~$A$ we have $\bigvee_n d_n' \aprx{A} e$, from which
    $\bigvee_n d_n \aprx{A \E \Drt} c$ follows because $\bigvee_n d_n = \inval (\bigvee_n
    d_n')$.
  \item there are $\iota$ and $\op$ such that for large enough $n$ there are $d_n'$ and
    $\kappa_n'$ such that $d_n = \inop{\hash{\iota}{\op}}(d_n', \kappa_n)$: this case is
    treated analogously to the previous one. \qedhere
  \end{itemize}
\end{proof}

\begin{lem}
  \label{lem:adequacy-general}%
  Let $\ctx$ be the context $x_1 : A_1, \ldots, x_n : A_n$. Suppose that for each $1 \leq
  i \leq n$ we have $d_i \in \skel{A_i}$ and a closed expression $\ent e_i \T A_i$ such
  that $d_i \aprx{A_i} e_i$.
  \begin{enumerate}
  \item If $\ctx \ent e : A$ then $\sem{\ctx \ent e : A}{(d_1, \ldots, d_n)} \aprx{A} e[e_1/x_1, \ldots, e_n/x_n]$.
  \item If $\ctx \ent c : \C$ then $\sem{\ctx \ent c : \C}{(d_1, \ldots, d_n)} \aprx{\C} c[e_1/x_1, \ldots, e_n/x_n]$.
  \end{enumerate}
\end{lem}

\begin{proof}
  We prove the statements by induction on the derivation of the typing judgments. Let
  $\eta = (d_1, \ldots, d_n)$ and $\sigma = [e_1/x_1, \ldots, e_n/x_n]$. We write
  $[\sigma, e/x]$ for the substitution $[e_1/x_1, \ldots, e_n/x_n, e/x]$. Throughout we
  assume that bound variables occurring in various terms do not appear in~$\sigma$:
  \begin{description}
    \item[Cases \rulename{Var}, \rulename{True}, \rulename{False}, \rulename{Unit}, \rulename{Zero}, \rulename{Inst}]
      these are all trivial.

    \item[Case \rulename{Succ}] $\ctx \ent (\succ e) : \natty$.
      By the induction hypothesis, we have $\sem{e}{\eta} \aprx{\natty} e \sigma$.
      Next, notice that the closed expression $e \sigma \T \natty$ must be $\kord{succ}^k \, 0$ for some $k$,
      hence $\sem{e}{\eta} = k$, and so
      \[
        \sem{\succ e}{\eta} = (k + 1) \aprx{\natty} \succ (e \sigma) = (\succ e) \sigma.
      \]

    \item[Case \rulename{Fun}] $\ctx \ent (\fun{x} c) : A \to \C$.
      If $d' \aprx{A} e'$ then by the induction hypothesis for~$c$
      \begin{equation*}
        \sem{\ctx, x : A \ent c}{(\eta, d')} \aprx{\C} c [\sigma, e'/x].
      \end{equation*}
      Because $\sem{\ctx, x : A \ent c}{(\eta, d')} = (\sem{\ctx \ent \fun{x} c}{\eta})\, d'$ and
      \begin{equation*}
        ((\fun{x} c)\sigma) \, e' = (\fun{x} c \sigma) \, e' \step c [\sigma, e'/x],
      \end{equation*}
      we may use Lemma~\ref{lem:approx-step} to get the desired conclusion
      \begin{equation*}
        (\sem{\ctx \ent \fun{x} c}{\eta})\, d' \aprx{\C} ((\fun{x} c)\sigma) \, e'.
      \end{equation*}
      Most other cases in the proof follow the same pattern, so we shall not explicitly
      mention uses of Lemma~\ref{lem:approx-step} anymore.

    \item[Case \rulename{Hand}] $\ctx \ent (\handler \val x \mapsto c_v \effcase \ocs) : A \E \Drt \hto B \E \Drt'$.
      We abbreviate the handler as $h$.
      Assuming $d \aprx{A \E \Drt} c$ we need to show that
      \begin{equation*}
        (\sem{h}{\eta})(d) \aprx{B \E \Drt'} (\withhandle{h \sigma}{c}).
      \end{equation*}
      We proceed by an induction on $d$:
      \begin{itemize}
      \item If $d = \bot$ then the conclusion follows because $\sem{h}{\eta}$ is strict.
      \item If $d = \inval(d')$ then $d \aprx{A \E \Drt} c$ implies $c \eval \val e'$ and $d' \aprx{A} e'$ for some $e'$ and $d'$.
      In this case, by induction hypothesis for~$c_v$,
        \begin{equation*}
          (\sem{h}{\eta})(d) = \sem{\ctx, x : A \ent c_v}{(\eta, d')} \aprx{b \E \Drt} c_v [\sigma, e'/x],
        \end{equation*}
        and
        \begin{equation*}
          (\withhandle{h \sigma}{c}) \step \cdots \step
          (c_v \sigma) [e'/x] = c_v [\sigma, e'/x].
        \end{equation*}
      \item If $d = \inop{\hash{\iota}{\op}}(d', \kappa)$ then $d \aprx{A \E \Drt} c$ implies $c \eval \call{\iota}{\op}{e'}{\cont{y}{c'}}$, $d' \aprx{A^\op} e'$
      and $\kappa \aprx{B^\op \to A \E \Drt} (\fun{y} c')$.
      Now
        \begin{equation*}
          (\sem{h}{\eta})(\inop{\hash{\iota}{\op}}(d', \kappa)) = \xsem{\ocs}_{\hash{\iota}{\op}} (\eta, d', \sem{h}{\eta} \circ \kappa)
        \end{equation*}
        and
        \begin{equation*}
          (\withhandle{h \sigma}{c}) \step \cdots \step (\ocs\, \sigma)_{\hash{\iota}{\op}} (e', (\fun{y} \withhandle{h \sigma}{c'})),
        \end{equation*}
        therefore it suffices to prove
        \begin{equation*}
          \xsem{\ocs}_{\hash{\iota}{\op}} (\eta, d', \sem{h}{\eta} \circ \kappa) \aprx{B \E \Drt'}
          (\ocs\, \sigma)_{\hash{\iota}{\op}} (e', (\fun{y} \withhandle{h \sigma}{c'})),
        \end{equation*}
        which we do by induction on the length of~$\ocs$. When $\ocs$ is $\ocsnil$ the statement becomes
        \begin{equation*}
          \inop{\hash{\iota}{\op}}(d', \sem{h}{\eta} \circ \kappa) \aprx{B \E \Drt'}
          \call{\iota}{\op}{e'}{\cont{y}{\withhandle{h \sigma}{c'}}}.
        \end{equation*}
        We already know $d' \aprx{A^\op} e'$,
        and still need
        \[
          \sem{h}{\eta}(\kappa(d'')) \aprx{B \E \Drt'} (\withhandle{h \sigma}{c'[e''/ /y]})
        \]
        assuming $d'' \aprx{B^\op} e''$. This follows from the available induction hypotheses.

        When $\ocs$ is $(\call{\iota'}{\op'}{x}{k} \mapsto c'' \effcase \ocs')$ there are two further subcases:
        \begin{itemize}
        \item if $\hash{\iota}{\op} = \hash{\iota'}{\op'}$ then we get
          \begin{align*}
            &\sem{\ctx, x, k \ent c''}{(\eta, d, \sem{h}{\eta} \circ \kappa)} \aprx{B \E \Drt'} {} \\
            &\quad c'' [\sigma, e'/x, (\fun{y}{\withhandle{h \sigma}{c'}})/k]
          \end{align*}
          which holds by induction on $c''$,
        \item if $\hash{\iota}{\op} \neq \hash{\iota'}{\op'}$ then we are left with
          \begin{align*}
            &\xsem{\ocs'}_{\hash{\iota}{\op}} (\eta, d', \sem{h}{\eta} \circ \kappa) \aprx{B \E \Drt'} \\
            &\quad (\ocs'\, \sigma)_{\hash{\iota}{\op}} (e, (\fun{y}{\withhandle{h \sigma}{c'}})),
          \end{align*}
          which is just the induction hypothesis for $\ocs'$.
        \end{itemize}
      \end{itemize}

    \item[Cases \rulename{IfThenElse}, \rulename{Match}, \rulename{Absurd}] 
      We consider only
      \[
        \ctx \ent (\conditional{e}{c_1}{c_2}) : \C,
      \]
      as \rulename{Match} is similar and \rulename{Absurd} is vacuous.
      Because $e \sigma$ is a closed expression of type $\boolty$ it is either $\tru$ or
      $\fls$. Let us take a look at the first possibility. By induction hypothesis for
      $c_1$ we have $\sem{c_1}{\eta} \aprx{\C} c_1 \sigma$. Again, because
      $\sem{\conditional{\tru}{c_1}{c_2}}{\eta} = \sem{c_1}{\eta}$ and
      \begin{equation*}
        (\conditional{\tru}{c_1}{c_2})\sigma \step c_1\sigma
      \end{equation*}
      it follows that
      \begin{equation*}
        \sem{\conditional{\tru}{c_1}{c_2}}{\eta} \aprx{\C} (\conditional{\tru}{c_1}{c_2})\sigma,
      \end{equation*}
      as required.

    \item[Case \rulename{App}] $\ctx \ent e_1' \, e_2' : \C$ where $\ctx \ent e_1' : A \to \C$
      and $\ctx \ent e_2' : A$.
      By induction hypotheses for $e_1'$ and $e_2'$ we have $\sem{e_1'}{\eta} \aprx{A \to
        \C} e_1' \sigma$ and $\sem{e_2'}{\eta} \aprx{A} e_2 \sigma$, therefore by the
      definition of $\aprx{A \to \C}$,
      \begin{equation*}
        \sem{e_1' \, e_2'}{\eta} =
        (\sem{e_1'}{\eta}) (\sem{e_2'}{\eta}) \aprx{\C} (e_1' \sigma) (e_2' \sigma) = (e_1' \, e_2') \sigma.
      \end{equation*}

    \item[Case \rulename{Val}] $\ctx \ent \val e : A \E \Drt$. 
      By induction on $e$ we have $\sem{e}{\eta} \aprx{A} e \sigma$, therefore by the
      second clause in the definition of $\aprx{A \E \Drt}$
      \begin{equation*}
        \sem{\val e}{\eta} = \inval(\sem{e}{\eta}) \aprx{A \E \Drt} \val (e \sigma) = (\val e) \sigma.
      \end{equation*}

    \item[Case \rulename{Op}] $\ctx \ent \call{\iota}{\op}{e}{\cont{y}{c}} : A \E \Drt$.
      This case works much like a combination of \rulename{Val} and \rulename{App}, so we omit the details.

    \item[Case \rulename{Let}] $\ctx \ent (\letin{x = c_1}{c_2}) : B \E \Drt$.
      This case is treated like a handler which only has a $\kord{val}$ case.

    \item[Case \rulename{With}] $\ctx \ent \withhandle{e}{c} : \D$ where $\ctx \ent e : \C \hto \D$
      and $\ctx \ent c : \C$. 
      By induction hypotheses we have $\sem{e}{\eta} \aprx{\C \hto \D} e\sigma$ and
      $\sem{c}{\eta} \aprx{\C} c\sigma$, therefore by the definition of~$\aprx{\C \hto
        \D}$
      \begin{equation*}
        (\sem{e}{\eta}) (\sem{c}{\eta})  \aprx{\D} (\withhandle{e \sigma}{c \sigma}).
      \end{equation*}
      The left-hand side equals $\sem{\withhandle{e}{c} : \D}{\eta}$ and the right
      $(\withhandle{e}{c}) \sigma$, so we are done.

    \item[Case \rulename{LetRec}] $\ctx \ent \letrecin{f x = c_1}{c_2} : \D$.
      After a short calculation the problem reduces to showing that
      \begin{equation*}
        g \aprx{A \to \C} (\fun{x} \letrecin{f x = c_1 \sigma} c_1 \sigma)
      \end{equation*}
      where $g$ is the least fixed-point of the operator $\Phi : \skel{A \to \C} \to \skel{A \to \C}$, defined by
      \begin{equation*}
        \Phi (h) = \lambda a \in \skel{A} \,.\, \sem{\ctx, f, x \ent c_1}{(\eta, h, a)}.
      \end{equation*}
      By Lemma~\ref{lem:approx-sup} it suffices to show that
      \begin{equation*}
        h \aprx{A \to \C} (\fun{x} \letrecin{f x = c_1 \sigma} c_1 \sigma)
      \end{equation*}
      implies
      \begin{equation*}
        \Phi(h) \aprx{A \to \C} (\fun{x} \letrecin{f x = c_1 \sigma} c_1 \sigma).
      \end{equation*}
      This amounts to proving that if $d \aprx{A} e$ then
      \begin{align*}
        \Phi(h)(d)
        &= \sem{\ctx, f, x \ent c_1}{(\eta, h, d)} \\
        &\aprx{\C} c_1 [\sigma, (\fun{x} \letrecin{f x = c_1 \sigma} c_1 \sigma)/f, e/x],
      \end{align*}
      which holds by the induction hypothesis for~$c_1$.

    \item[Case \rulename{SubExpr}, \rulename{SubComp}]
      For \rulename{SubExpr}, take $\ctx \ent e \T A'$ where $\ctx \ent e \T A$ and $A \le A'$.
      From induction hypothesis, we get $\sem{\ctx \ent e \T A}{\eta} \aprx{A} e \sigma$.
      Since $\skel{A} = \skel{A'}$, we have $\aprx{A} = \aprx{A'}$.
      Additionally, $\sem{\ctx \ent e \T A'} = \sem{\ctx \ent e \T A}$,
      hence $\sem{e}{\eta} \aprx{A'} e \sigma$.
      For \rulename{SubComp}, the proof is similar. \qedhere
    \end{description}
\end{proof}

\begin{cor}[Adequacy]
  \label{cor:adequacy}%
  If $\ent c : \unitty \E \Drt$ and $\xsem{c} = \inval(\star)$ then $c \eval \val \unt$.
\end{cor}

\begin{proof}
  By the previous lemma $\xsem{c} \aprx{\unitty \E \Drt} c$. Therefore, if $\xsem{c} =
  \inval(\star)$ then $c \eval \val \unt$ by the definition of $\aprx{\unitty \E \Drt}$.
\end{proof}

The stated adequacy suffices for our purposes, but of course similar statements holds for
other ground types. Regarding operations, if $\xsem{c} = \inop{\hash{\iota}{\op}}(d,
\kappa)$, then $\hash{\iota}{\op} \in \Drt$ and so $c \eval \call{\iota}{\op}{e}{k}$ for
some $e$ and $k$ such that $\xsem{e} = d$ and $\xsem{k} = \kappa$. Therefore, if $\xsem{c}
\neq \bot$ then $c \eval r$ for some result~$r$.

%%% Local Variables: 
%%% mode: latex
%%% TeX-master: "effect-system"
%%% End: 

\section{Equational reasoning}
\label{sec:reasoning}

\subsection{Contextual and denotational equivalence}
\label{sub:cont-denot-equiv}

In this section we provide principles that allow us to reason about programs. Let us first
recall how \emph{contextual equivalence} is defined. An \emph{expression context}~$\evctx$
is a computation with several occurrences of a \emph{hole} $[\,]$ in positions where an
expression is expected. When the hole is plugged with an expression $e$ we get a computation
$\evctx[e]$. A \emph{computation context}~$\evctx$ is defined analogously, except that the
holes appear where computations are expected.

We say that expressions $e$ and $e'$ are \emph{contextually equivalent}, written $e
\approx e'$, when for all expression contexts $\evctx$ such that $\evctx[e]$ and
$\evctx[e']$ are both of type $\unitty \E \Drt$, we have $\evctx[e] \eval \val \unt$ if and only
if $\evctx[e'] \eval \val \unt$. Contextual equivalence of computations is defined
analogously.

Contextually equivalent expressions or computations may be interchanged anywhere in the
code. Therefore, it is quite useful to know that a certain contextual equivalence holds,
but unfortunately it is difficult to work directly with contextual equivalence. Luckily,
denotational equivalence is more easily handled and is related to contextual equivalence
by the adequacy theorem.

We write $e \equiv e'$ and $c \equiv c'$ when the denotations of two expressions or
computations are the same. More precisely, if $\ctx \ent e : A$ and $\ctx \ent e' : A$
then $\ctx \ent e \equiv e' : A$, or just $e \equiv e'$, means $\xsem{\ctx \ent e \T A} =
\xsem{\ctx \ent e' \T A}$, and similarly for computations.

\begin{prop}
  Denotationally equal expressions are contextually equivalent, and likewise for computations.
\end{prop}

\begin{proof}
  Suppose $e \equiv e'$ and consider an expression context $\evctx$ such that $\evctx[e]$
  and $\evctx[e']$ both have type $\unitty \E \Drt$. Assume $\evctx[e] \eval \val \unt$.
  By soundness of denotational semantics $\xsem{\evctx[e]} = \inval (\star)$. By assumption
  $\xsem{e} = \xsem{e'}$ and because denotational semantics is compositional it follows that
  $\xsem{\evctx[e]} = \xsem{\evctx[e']}$. Now by adequacy $\evctx[e'] \eval \val \unt$.
  The proof for computations is the same.
\end{proof}

\indent Denotational equivalence is a congruence and is preserved by well-typed
substitutions. It validates the following $\beta$-rules, where the various expressions and
computations have suitable types, and $h$ stands for $\handler \val x \mapsto c_v \effcase
\ocs$:
\begin{align*}
  \conditional{\tru}{c_1}{c_2} &\Equiv c_1 \\
  \conditional{\fls}{c_1}{c_2} &\Equiv c_2 \\
  \matchnat{0}{c_1}{x}{c_2} &\Equiv c_1 \\
  \matchnat{(\succ e)}{c_1}{x}{c_2} &\Equiv c_2[e / x] \\
  (\fun{x} c) \, e &\Equiv c[e / x] \\
  \letin{x = \val e}{c} &\Equiv c[e / x] \\
  \letin{x = \call{e_1}{\op}{e_2}{\cont{y}{c_1}}} c_2
    &\Equiv \call{e_1}{\op}{e_2}{\cont{y}{\letin{x = c_1} c_2}} \\
  \letrecin{f \, x = c_1} c_2
    &\Equiv c_2[(\fun{x} \letrecin{f \, x = c_1} c_1) / f] \\
  \withhandle{h}{(\val e)} &\Equiv c_v[e / x] \\
  \withhandle{h}{(\call{\iota}{\op}{e}{\cont{y}{c}})}
    &\Equiv \ocs_{\hash{\iota}{\op}}(e, (\fun{y} \withhandle{h}{c}))
\end{align*}
We also have the following $\eta$-rules, provided the expressions and computations
have easily guessed types:
\begin{align*}
  e &\Equiv \unt \\
  \fun{x} e \, x &\Equiv  e \\
  \letin{x = c}{\val x} &\Equiv c \\
  \conditional{e}{c[\tru / x]}{c[\fls / x]} &\Equiv c[e / x] \\
  \matchnat{e}{c[0 / x]}{y}{c[\succ y / x]} &\Equiv c[e / x] \\
  \absurd{e} &\Equiv c[e / x]
\end{align*}
We omit the proofs because they just involve unfolding of semantic definitions. An
exception is the $\eta$-rule for $\kord{let}$ binding, which is proved by induction in the
next section.

A variety of other equivalences is readily validated, for example
\[
  \letin{x = c_1} c_2 \Equiv \withhandle{(\handler \val x \mapsto c_2 \effcase \ocsnil)}{c_1}
\]
and the ``associativity'' of $\kord{let}$ binding~\cite{DBLP:journals/iandc/Moggi91}
\[
  \letin{x = (\letin{y = c_1} c_2)} c_3
  \Equiv
  \letin{y = c_1}{(\letin{x = c_2} c_3)},
\]
where $y$ must not occur freely in $c_3$, see~\cite{plotkin08a-logic, pretnar10the-logic}
for other examples.

\subsection{An induction principle for effects}

The induction principle for the computation domains from~Section~\ref{sub:computation-domains} is
useful for deriving general laws that do not depend on a particular choice of effects. It
is less useful for specific examples which involve a carefully chosen set of effects and
handlers, because it forces us to consider operations that have nothing to do with the
situation at hand. The induction principle Theorem~\ref{thm:tsem-induction} remedies the drawback.

A typical application arises when we prove an equivalence of the form
$\evctx[c] \equiv \evctx'[c] : \C$ for all computations $c$ of a suitable type $A \E
\Drt$. The computation contexts $\evctx$ and $\evctx'$ are interpreted as continuous maps
$\xsem{\evctx}, \xsem{\evctx'} : \skel{A \E \Drt} \to \skel{\C}$, and the equivalence may
be phrased as
\begin{equation*}
  \fra{t \in \tsem{A \E \Drt}} \sem{\evctx}(t) = \sem{\evctx'}(t).
\end{equation*}
The equality inside the quantifier is an admissible predicate on $\tsem{A \E \Drt}$, so
the induction principle applies. It is a bit cumbersome to perform the proof using the
semantic brackets $\xsem{{-}}$ all over the place. Instead, with a bit of flexibility in
notation, we can write the proof in a syntactic manner as follows:
\begin{enumerate}
\item We write $\bot$ for a non-terminating computation, e.g., $\letrecin{f x = f x}{f
    \unt}$, and check that $\evctx[\bot] \equiv \evctx'[\bot]$.
\item We verify that $\evctx[\val e] \equiv \evctx'[\val e]$ where $e$ is a meta-variable of type $A$.
\item We verify, for each $\hash{\iota}{\op} \in \Drt$,
  \begin{equation*}
    \evctx[\call{\iota}{\op}{e}{\cont{y}{\kappa\, y}}] \equiv \evctx'[\call{\iota}{\op}{e}{\cont{y}{\kappa\, y}}],
  \end{equation*}
  where $e$ is a meta-variable of type $A^\op$ and $\kappa$ is a meta-variable of type
  $B^\op \to A \E \Drt$. The induction hypothesis is that $\evctx[\kappa \, e'] \equiv
  \evctx'[\kappa \, e']$ for all expressions $e'$ of type $B^\op$.
\end{enumerate}
We apply the method to prove the $\eta$-rule
\begin{equation*}
  \letin{x = c} \val x \Equiv c
\end{equation*}
by induction:
\begin{enumerate}
\item $(\letin{x = \bot} \val x) \equiv \bot$ because $\kord{let}$ is strict in both arguments,
\item $(\letin{x = (\val a)} \val x) \equiv (\val x)[a/x] \equiv \val a$ by a $\beta$-rule for $\kord{let}$,
\item The induction step is proved by
  \begin{align*}
    & \letin{x = (\call{\iota}{\op}{e}{\cont{y}{\kappa \, y}})} \val x \\
    &\quad
    \begin{aligned}
      &\Equiv
      \call{\iota}{\op}{e}{\cont{y}{
        \letin{x = \kappa \, y} \val x
      }} \\
      &\Equiv
      \call{\iota}{\op}{e}{\cont{y}{
        \kappa \, y
      }}
    \end{aligned}
  \end{align*}
  where we used a $\beta$-rule in the first step and the induction hypothesis for $\kappa \, y$ in the second.
\end{enumerate}
A non-trivial application of the induction principle is presented in Section~\ref{sub:comm-non-interf}.

%%% Local Variables: 
%%% mode: latex
%%% TeX-master: "effect-system"
%%% End: 

\section{Example: mutable references}
\label{sec:examples}

As a more elaborate example we consider mutable references.
In core \Eff they are implemented by the effect $\kord{ref}$ with operations
\begin{equation*}
  \kord{lookup} : \unitty \to A
  \qquad\text{and}\qquad
  \kord{update} : A \to \unitty,
\end{equation*}
where $A$ is a fixed pure type (in full \Eff we could use a parameter).
The handler which handles the reference given by an expression $r \T \kord{ref}^\Rgn$ is defined by
(the underscore $\anon$ indicates an ignored parameter):
\begin{align*}
  \kord{state}_r \defeq
  \handler
  &\val x \mapsto \val (\fun{s} \val x) \\
  \effcase {}& \call{r}{\kord{lookup}}{\anon}{k} \mapsto \val (\fun{s} \letin{f = k \, s} f \, s) \\
  \effcase {}& \call{r}{\kord{update}}{s'}{k} \mapsto \val (\fun{s} \letin{f = k \, \unt} f \, s')
\end{align*}
This is just the usual monadic-style treatment of state which wraps a computation into
a state-carrying function. For any instance $\iota \in \iotas_{\kord{ref}}$, pure type $B$,
and dirt $\Drt$, the computation $\kord{state}_\iota$ has the type
\begin{equation*}
  B \E (\set{\hash{\iota}{\kord{lookup}},\hash{\iota}{\kord{update}}} \cup \Drt)
  \hto
  (A \to B \E \Drt) \E \Drt.
\end{equation*}
The type says that the handler erases lookups and updates of $\iota$ from
a computation. The return type $(A \to B \E \Drt) \E \Drt$ has an outer dirt $\Drt$
because effects could happen before the first lookup or update.
Note that the double dirt would not arise if we took
the call-by-push-value approach to handlers~\cite{plotkin13handling}.

The handler wraps a
handled computation $c$ of type $B \E \Drt$ into a function expecting the current state,
which explains the type $A \to B \E \Drt$. In practice such a function is immediately
applied to an initial state~$e$ of type~$A$ (such a final transformation of handled
computations is so common that in full \Eff handlers have a special $\kord{finally}$ case
just for this purpose):
\begin{equation*}
  \letin{f = (\withhandle{h}{c})} f \, e
\end{equation*}
The type of this computation is simply $B \E \Drt$. In particular, if the only effects
in $c$ are lookups and updates of $\iota$, we get a pure computation.

If the type of $r$ is weaker, for example $\kord{ref}^{\set{\iota_1, \iota_2}}$,
we are not able to deduce anything useful. The best we can do is
\begin{equation*}
  \kord{state}_r \T B \E \Drt \hto (A \to B \E \Drt) \E \Drt
\end{equation*}
for any dirt $\Drt$.

We may handle several references at once by wrapping a computation into several handlers.
For example, let $c$ be the computation which swaps the contents of two references:
\begin{align*}
  &\letin{y_1 = \hash{\iota_1}{\kord{lookup}} \, \unt} \\
  &\letin{y_2 = \hash{\iota_2}{\kord{lookup}} \, \unt} \\
  &\letin{\anon = \hash{\iota_1}{\kord{update}} \, y_2} \\
  &\letin{\anon = \hash{\iota_2}{\kord{update}} \, y_1} (\val \unt)
\end{align*}
By itself, $c$ has the type
\begin{equation*}
  \unitty \E \set{
    \hash{\iota_1}{\kord{lookup}}, \hash{\iota_1}{\kord{update}},
    \hash{\iota_2}{\kord{lookup}}, \hash{\iota_2}{\kord{update}}
  },
\end{equation*}
When $c$ is wrapped by the handler $h_2 = \kord{state}_{\iota_2}$,
\[
  \letin{f_2 = (\withhandle{h_2}{c})} f_2 \, e_2,
\]
the type becomes
$\unitty \E \set{\hash{\iota_1}{\kord{lookup}}, \hash{\iota_1}{\kord{update}}}$.
When the handler $h_1 = \kord{state}_{\iota_1}$ is used on top of that,
\[
  \letin{f_1 = \big(\withhandle{h_1}{\big(\letin{f_2 = (\withhandle{h_2}{c})} f_2 \, e_2\big)}\big)} f_1 \, e_1
\]
we get the pure type $\unitty \E \purely$. Beware, it is important that the state is initialized at the correct point in the computation. For instance,
\[
  \letin{f_1 = \big(\withhandle{h_1}{(\withhandle{h_2}{c})}\big)} 
  (\letin{f_2 = f_1 \, e_1} f_2 \, e_2)
\]
does not do the right thing. We are warned about possible trouble by the effect system
which gives the computation the type $\unitty \E \set{\hash{\iota_1}{\kord{lookup}},
  \hash{\iota_1}{\kord{update}}}$ --- the operations for $\iota_1$ are escaping the handlers!
A less modular way of handling two instances is to create a new handler with four operation cases, two for each of the instances.

\subsection{Reasoning about references}
\label{sub:reasoning-references}

With handlers the workings of a computation may be inspected in a highly intensional way.
Consequently, there are few generally valid observational equivalences. However, when known
handlers are used to handle operations, we may derive equivalences that describe the
behavior of operations. The situation is opposite to that of~\cite{plotkin13handling},
where we start with an equational theory for operations and require that the handlers
respect it.

We demonstrate the technique for mutable state. Let $h = \kord{state}_\iota$ and abbreviate 
\begin{equation*}
  \letin{f = (\withhandle{h}{c})} f \, e
\end{equation*}
as $\hndl[c, e]$. Straightforward calculations give us the equivalences
\begin{align*}
  \hndl[(\call{\iota}{\kord{lookup}}{\unt}{\cont{y}{c}}), e]
  &\Equiv
  \hndl[c[e / y], e] \\
  \hndl[(\call{\iota}{\kord{update}}{e'}{\cont{\anon}{c}}), e]
  &\Equiv
  \hndl[c, e'] \\
  \hndl[\val e', e]
  &\Equiv
  \val e',
\end{align*}
for instance,
\begin{align*}
  &\hndl[(\call{\iota}{\kord{update}}{e'}{\cont{\anon}{c}}), e] \\
  &\qquad
  \begin{aligned}[t]
    &\Equiv
    \letin{f = \val (\fun{s}
      \letin{f' = (\fun{\anon} \withhandle{h}{c}) \, \unt} f' \, e'
      )} f \, e
    \\ 
    &\Equiv
    \letin{f = \val (\fun{s} \hndl[c, e'])} f \, e
    \\
    &\Equiv
    (\fun{s} \hndl[c, e']) \, e
    \\
    &\Equiv
    \hndl[c, e'].
  \end{aligned}
\end{align*}
These suffice for simple equational reasoning about state. If we read them as rewrite
rules they allow us to progressively transform a computation to a simpler form. In fact,
the transformations mimic the usual coalgebraic operational semantics for state~\cite{plotkin08:_tensor_comod_model_operat_seman}.

Of course, a realistic computation will contain several handlers. As long as they do not
interfere with each other, we can still use equivalences to usefully manipulate them. For
example, if $h'$ is a handler with no operation case for $\hash{\iota}{\kord{lookup}}$ then
\begin{align*}
  & \hndl [
    (\withhandle{h'}{\call{\iota}{\kord{lookup}}{\unt}{\cont{y}{c}}}),
    e] \\
  &\qquad
  \begin{aligned}[t]
    &\Equiv
    \hndl[
      (\call{\iota}{\kord{lookup}}{\unt}{\cont{y}{\withhandle{h'}{c}}}),
      e]
    \\
  &\Equiv
  \hndl [(\withhandle{h'}{c[e / y]}), e].
\end{aligned}
\end{align*}
It may happen that a lookup or an update is nested deeply inside several handlers. The
above transformation allows us to hoist the operation out of the inner handlers so that it
is handled by the outer handler, as long as the inner handlers do not attempt to
handle~$\iota$. The transformation applies to $\kord{let}$ bindings too, as they are like
handlers without operation cases.

We may validate the seven standard equations governing state~\cite{plotkin02notions}.
There are four combinations of lookup and update (in the first equation $y$ does not occur
free in $c$):
\begin{align*}
  \hndl [\call{\iota}{\kord{lookup}}{\unt}{\cont{y}{\call{\iota}{\kord{update}}{y}{\cont{\anon}{c}}}}, e]
    &\Equiv \hndl [c, e] \\
  \hndl [\call{\iota}{\kord{lookup}}{\unt}{\cont{y}{\call{\iota}{\kord{lookup}}{\unt}{\cont{z}{c}}}}, e]
    &\Equiv \hndl [\call{\iota}{\kord{lookup}}{\unt}{\cont{y}{c[y / z]}}, e] \\
  \hndl [\call{\iota}{\kord{update}}{e}{\cont{\anon}{\call{\iota}{\kord{update}}{e'}{\cont{\anon}{c}}}}, e]
    &\Equiv \hndl[ \call{\iota}{\kord{update}}{e'}{\cont{\anon}{c}}, e] \\
  \hndl [\call{\iota}{\kord{update}}{e}{\cont{\anon}{\call{\iota}{\kord{lookup}}{\unt}{\cont{y}{c}}}}, e]
    &\Equiv \hndl [\call{\iota}{\kord{update}}{e}{\cont{\anon}{c[e / y]}}, e]
\end{align*}
For instance, the first equation is validated as follows:
\begin{align*}
  &\hndl [
    \call{\iota}{\kord{lookup}}{\unt}{\cont{y}{\call{\iota}{\kord{update}}{y}{\cont{\anon}{c}}}},
    e]
  \\
  &\qquad
  \begin{aligned}[t]
  &\Equiv \hndl [\call{\iota}{\kord{update}}{e}{\cont{\anon}{c}}, e]
  \\
  &\Equiv \hndl[c, e]
  \end{aligned}
\end{align*}
Three more equations describe commutativity of lookups and updates at different distances. Let $\iota_1 \neq \iota_2$, and write $\hndl_1$ and $\hndl_2$ for the the abbreviation $\hndl$ with respect to $\iota_1$ and $\iota_2$, respectively:
\begin{align*}
  &\hndl_1[\hndl_2[
  \call{\iota_1}{\kord{lookup}}{\unt}{\cont{y_1}{\call{\iota_2}{\kord{lookup}}{\unt}{\cont{y_2}{c}}}}
   ,e_2], e_1] \\
   &\qquad
   \Equiv
   \hndl_1[\hndl_2[
    \call{\iota_2}{\kord{lookup}}{\unt}{\cont{y_2}{\call{\iota_1}{\kord{lookup}}{\unt}{\cont{y_1}{c}}}}
   ,e_2], e_1]
   \\
  &\hndl_1[\hndl_2[
  \call{\iota_1}{\kord{update}}{e_1}{\cont{\anon}{\call{\iota_2}{\kord{update}}{e_2}{\cont{\anon}{c}}}}
  ,e_2] , e_1] \\
    &\qquad
    \Equiv
   \hndl_1[\hndl_2[
   \call{\iota_2}{\kord{update}}{e_2}{\cont{\anon}{\call{\iota_1}{\kord{update}}{e_1}{\cont{\anon}{c}}}}
   , e_2], e_1] \\
  &\hndl_1[\hndl_2[
    \call{\iota_1}{\kord{update}}{e}{\cont{\anon}{\call{\iota_2}{\kord{lookup}}{\unt}{\cont{y_2}{c}}}}
    , e_2], e_1] \\
  &\qquad
  \Equiv
  \hndl_1[\hndl_2[
  \call{\iota_2}{\kord{lookup}}{\unt}{\cont{y_2}{\call{\iota_1}{\kord{update}}{e}{\cont{\anon}{c}}}}
  , e_2], e_1]
\end{align*}
Let us check the last equation. The left-hand side transforms as
\begin{align*}
  &\hndl_1 [
    \hndl_2 [
      \call{\iota_1}{\kord{update}}{e}{\cont{\anon}{\call{\iota_2}{\kord{lookup}}{\unt}{\cont{y_2}{c}}}}
    , e_2]
  , e_1]
  \\
  &\qquad
  \begin{aligned}[t]
    &\Equiv
    \hndl_1 [
    \call{\iota_1}{\kord{update}}{e}{\cont[]{\anon}{
      \hndl_2 [
        \call{\iota_2}{\kord{lookup}}{\unt}{\cont{y_2}{c}}
      , e_2]
    }}
  , e_1]
  \\
  &\Equiv
  \hndl_1 [
    \hndl_2 [
      \call{\iota_2}{\kord{lookup}}{\unt}{\cont{y_2}{c}}
    , e_2]
  , e]
  \\
  &\Equiv
  \hndl_1 [
    \hndl_2[c[e_2 / y_2], e_2]
  , e]
  \end{aligned}
\intertext{and the right-hand side as}
  &\hndl_1 [
    \hndl_2 [
      \call{\iota_2}{\kord{lookup}}{\unt}{\cont{y_2}{\call{\iota_1}{\kord{update}}{e}{\cont{\anon}{c}}}}
    ,e_2 ]
  , e_1]
  \\
  &\qquad
  \begin{aligned}[t]
  & \Equiv {}
  \hndl_1 [
    \hndl_2 [
      \call{\iota_1}{\kord{update}}{e}{\cont{\anon}{c[e_2 / y_2]}}
    , e_2]
  , e_1]
  \\
  &\Equiv
  \hndl_1 [
    \call{\iota_1}{\kord{update}}{e}{\cont[]{\anon}{
      \hndl_2[c[e_2 / y_2], e_2]
    }}
  , e_1]
  \\
  &\Equiv
  \hndl_1 [
      \hndl_2[c[e_2 / y_2], e_2]
  , e].
  \end{aligned}
\end{align*}
The remaining two equations are proved much the same way. The symmetry in the equations
also shows that it does not matter in which order we nest the handlers for $\iota_1$ and $\iota_2$.

\subsection{Commutativity of non-interfering computations}
\label{sub:comm-non-interf}

Swapping lookups and updates that act on different instances is only a basic reasoning step. In
practice we want to swap whole computations, as long as they do not interfere. To make the idea precise,
let $\Drt_1 = \set{\hash{\iota_1}{\kord{lookup}}, \hash{\iota_1}{\kord{update}}}$, 
$\Drt_2 = \set{\hash{\iota_2}{\kord{lookup}}, \hash{\iota_2}{\kord{update}}}$,
let $c_1$ and $c_2$ be computations of types $A_1 \E \Drt_1$ and $A_2 \E \Drt_2$, respectively,
and $c$ a computation of type $\C$ in the context $x_1 : A_1, x_2 : A_2$.
We would like to show the equivalence
\begin{align*}
   &\hndl_1[\hndl_2[
    \letin{x_1 = c_1}
    (\letin{x_2 = c_2} c), e_2], e_1] \\
  &\qquad \Equiv \hndl_1[\hndl_2[
    \letin{x_2 = c_2}
    (\letin{x_1 = c_1} c), e_2], e_1].
\end{align*}
This is a commutativity law which allows us to transpose, or run in parallel, any
computations that use only non-interfering references.

First, let us establish a simpler equivalence, which we are going to use in the proof:
\begin{align*}
  &\hndl_1[\hndl_2[\call{\iota_1}{\kord{lookup}}{\unt}{\cont{y_1}{
    \letin{x_2 = c_2}
    (\letin{x_1 = c_1'} c)}}, e_2], e_1] \\
  &\qquad \Equiv
  \hndl_1[\hndl_2[
    \letin{x_2 = c_2}
    (\letin{x_1 = \call{\iota_1}{\kord{lookup}}{\unt}{\cont{y_1}{c_1'}}} c), e_2], e_1]
\end{align*}
We proceed by induction on $c_2$.
Since handlers and $\kord{let}$ are strict,
both sides are $\bot$ when we set $c_2$ to $\bot$.
Next, if $c_2$ is $\val e_2$, the two sides are equal already without handlers:
\begin{align*}
  &\call{\iota_1}{\kord{lookup}}{\unt}{\cont{y_1}{
    \letin{x_2 = \val e_2}
    (\letin{x_1 = c_1'} c)}} \\
  &\qquad \Equiv
  \call{\iota_1}{\kord{lookup}}{\unt}{\cont{y_1}{
    \letin{x_1 = c_1'} c[e_2 / x_2]}} \\
  &\qquad \Equiv
  \call{\iota_1}{\kord{lookup}}{\unt}{\cont{y_1}{
    \letin{x_1 = c_1'}
    (\letin{x_2 = \val e_2} c)}} \\
  &\qquad \Equiv
    \letin{x_1 = \call{\iota_1}{\kord{lookup}}{\unt}{\cont{y_1}{c_1'}}}
    (\letin{x_2 = \val e_2} c)
\end{align*}
For the induction step, suppose $c_2$ is a call of $\hash{\iota_2}{\kord{update}}$ (the case $\kord{lookup}$ is similar):
\begin{align*}
  &\hndl_1[\hndl_2[\call{\iota_1}{\kord{lookup}}{\unt}{\cont{y_1}{
  \letin{x_2 = \call{\iota_2}{\kord{update}}{e_2'}{\cont{z}{\kappa \, z}}}
  (\letin{x_1 = c_1'} c)}}, e_2], e_1] \\
  &\quad \Equiv
  \hndl_1[\call{\iota_1}{\kord{lookup}}{\unt}{\cont{y_1}{\hndl_2[
  \letin{x_2 = \kappa \unt}
  (\letin{x_1 = c_1'} c), e_2']}}, e_1] \\
  &\quad \Equiv
  \hndl_1[\hndl_2[\call{\iota_1}{\kord{lookup}}{\unt}{\cont{y_1}{
  \letin{x_2 = \kappa \unt}
  (\letin{x_1 = c_1'} c)}}, e_2'], e_1] \\
  &\quad \Equiv
  \hndl_1[\hndl_2[
  \letin{x_2 = \kappa \unt}
  (\letin{x_1 = \call{\iota_1}{\kord{lookup}}{\unt}{\cont{y_1}{c_1'}}} c), e_2'], e_1] \\
  &\quad \Equiv
  \hndl_1[\hndl_2[
  \letin{x_2 = \call{\iota_2}{\kord{update}}{e_2'}{\cont{z}{\kappa \, z}}}
  (\letin{x_1 = \call{\iota_1}{\kord{lookup}}{\unt}{\cont{y_1}{c_1'}}} c), e_2], e_1].
\end{align*}
In the third step we used the induction hypothesis for $\kappa \unt$.

We now prove the main statement by induction on $c_1$.
When $c_1$ is $\bot$ or a value, the statement is easy.
If $c_1$ is a call of $\hash{\iota_1}{\kord{lookup}}$, we have:
\begin{align*}
  &\hndl_1[\hndl_2[
    \letin{x_1 = \call{\iota_1}{\kord{lookup}}{\unt}{\cont{y_1}{\kappa \, y_1}}}
    (\letin{x_2 = c_2} c), e_2], e_1] \\
  &\qquad \Equiv
  \hndl_1[\call{\iota_1}{\kord{lookup}}{\unt}{\cont{y_1}{\hndl_2[
    \letin{x_1 = \kappa \, y_1}
    (\letin{x_2 = c_2} c), e_2}}], e_1] \\
  &\qquad \Equiv
  \hndl_1[\hndl_2[
    \letin{x_1 = \kappa \, e_1}
    (\letin{x_2 = c_2} c), e_2], e_1] \\
  &\qquad \Equiv
  \hndl_1[\hndl_2[
    \letin{x_2 = c_2}
    (\letin{x_1 = \kappa \, e_1} c), e_2], e_1]
\end{align*}
where in the last step, we used the induction hypothesis for $\kappa \, e_1$.
The last line is equivalent to the other side of the desired equivalence:
\begin{align*}
  &\hndl_1[\hndl_2[
    \letin{x_2 = c_2}
    (\letin{x_1 = \kappa \, e_1} c), e_2], e_1] \\
  &\qquad \Equiv
  \hndl_1[\call{\iota_1}{\kord{lookup}}{\unt}{\cont{y_1}{\hndl_2[
    \letin{x_2 = c_2}
    (\letin{x_1 = \kappa \, y_1} c), e_2]}}, e_1] \\
  &\qquad \Equiv
  \hndl_1[\hndl_2[
    \letin{x_2 = c_2}
    (\letin{x_1 = \call{\iota_1}{\kord{lookup}}{\unt}{\cont{y_1}{\kappa \, y_1}}} c), e_2], e_1].
\end{align*}
The proof for $\kord{update}$ is similar.

\section{Formalization in Twelf}
\label{sec:formalization}

We formalized core \Eff, the effect system and the safety theorems in Twelf. The files are
enclosed with this paper, or can be found at the GitHub
repository~\cite{bauer2013formalization}. The code is compatible with Twelf version 1.7.1.
Further instructions and description of the code can be found in the file
\texttt{README.md}.

\section{Discussion}
\label{sec:discussion}

The only essential feature of \Eff that is missing from core \Eff is dynamic creation of
instances with the $\kpre{new} E$ construct. We omitted it because it leads to significant
complications, both in the effect system and in semantics. One possible treatment of
$\kord{new}$ would be to upgrade the current setup with nominal logic and nominal domain
theory.

Non-termination is a computational effect which is not reflected in our effect system. It
would be interesting to add a ``divergence'' effect. Such an effect would originate from
applications of recursive functions. However, since divergence cannot be handled it would
never disappear from effect information, and would likely become an uninformative
nuisance. A potential remedy would be to separate general recursive definitions, which may
diverge, from structural recursive definitions, which always terminate.

A realistic implementation of our effect system would only be useful if it actually
\emph{inferred} computational effects. Such a possibility was explored by the
second author~\cite{Pretnar13}, and an early prototype is available in the latest
implementation of \Eff~\cite{bauer14:_eff}.

\subsection*{Acknowledgment}

We thank the anonymous referees for useful suggestions and a lesson in domain theory.

\label{sec:acknowledgment}

%%% Local Variables: 
%%% mode: latex
%%% TeX-master: "effect-system"
%%% End: 

\bibliographystyle{plain}
\bibliography{bibliography}

\end{document}